\documentclass[journal,11pt,onecolumn]{IEEEtran}
\usepackage{amsmath,amsthm,amssymb,amscd}
\usepackage{amsfonts}
\usepackage{graphics}
\usepackage{epsfig}             
\usepackage{comment}
\usepackage{color}

\newtheorem{theorem}{Theorem}
\newtheorem{corollary}{Corollary}
\newtheorem{lemma}{Lemma}
\newtheorem{definition}{Definition}

\def\eps{\varepsilon}

\def\CR{{\mathcal R}}

\def\CG{{\mathcal G}}

\def\CX{{\mathcal X}}
\def\CY{{\mathcal Y}}
\def\E{{\bf E}}

\def\A{{\mathcal A}}
\def\S{{\mathcal S}}
\def\th{\theta}

\def\vpi{\mbox{\boldmath $\pi$}}

\def\br{{\rm r}}
\def\la{\langle}
\def\ra{\rangle}
\def\bx{{\bf x}}
\def\by{{\bf y}}
\def\bu{{\bf u}}

\def\bc{{\bf c}}
\def\bun{{\bf u}}
\def\fps#1{\frac{\partial}{\partial s_#1}}

\def\bP{{\bf P}}
\def\bQ{{\bf Q}}

\def\bI{{\bf I}}
\def\bM{{\bf M}}
\def\bC{{\bf C}}

\def\bpi{{\boldsymbol{\pi}}}
\def\bG{{\bf G}}
\def\br{{\bf r}}

\def\eps{\epsilon}
\def\CA{\A}

\def\CG{{\cal G}}
\def\binomleft{\left(\begin{array}{c}}
\def\binomright{\end{array}\right)}
\def\la{\langle}
\def\ra{\rangle}
\def\be{\begin{equation}}
\def\ee{\end{equation}}
\def\beq{\begin{eqnarray}}
\def\eeq{\end{eqnarray}}
\def\beqs{\begin{eqnarray*}}
\def\eeqs{\end{eqnarray*}}

\def\fp2u{\frac{\partial^2}{\partial u^2}}

\def\CU{{\cal U}}
\def\CK{{\cal K}}
\def\CR{{\cal R}}

\def\CI{{\cal I}}

\def\bc{{\bf c}}

\def\vpi{\mbox{\boldmath $\pi$}}
\def\bzeta{\mbox{\boldmath $\zeta$}}
\def\lb{\lambda}

\def\parsec{\par\noindent}

\def\med{\medskip\parsec}
\def\E{\mathbf{E}}

\def\CB{{\cal B}}
\def\bbe{{\bf e}}
\def\Res{\hbox{\rm Res}}
\def\td{\bar{d}}

\def\bone{{\bf 1}}

\def\CR{{\cal R}}

\def\bD{{\bf D}}
\def\tbC{\widetilde{\bC}}
\def\hbC{\widehat{\bC}}
\def\hC{\widehat{C}}

\def\CW{{\mathcal W}}

\def\trace{\hbox{\rm trace}}

\begin{document}
\title{Joint String Complexity for Markov Sources: \\
Small Data Matters}

\author{
\IEEEauthorblockN{Philippe Jacquet and Dimitris Milioris and
Wojciech Szpankowski,  \IEEEmembership{Fellow, IEEE,}}
\thanks{
Parts of this paper were presented
at the 2013 IEEE International Symposium on Information Theory, Istanbul,
Turkey.

Philippe Jacquet, Bell Labs - Nokia  91620 Nozay,
France (e-mail: Philippe.Jacquet@inria.fr).

Dimitris Milioris, Bell Labs - Nokia, 91620 Nozay,
France (e-mail: milioris@alcatel-lucent.com.

W. Szpankowski is with the Department of Computer Science, Purdue
University, IN 47907, USA (e-mail: spa@cs.purdue.edu);
also with the Faculty of Electronics, Telecommunications and
Informatics,  Gda\'{n}sk University of Technology, Poland.
This work was supported by NSF Center for Science of Information (CSoI) 
Grant CCF-0939370, and in addition by NSF Grants CCF-1524312, and
NIH Grant 1U01CA198941-01.

}}

\maketitle
\thispagestyle{plain}

\begin{abstract}
String complexity is defined as the cardinality of a set of all distinct words
(factors) of a given string.  For two strings, we introduce
the {\it joint string complexity} as the cardinality of a set of words that are 
common to both strings. String complexity finds a number of applications 
from capturing the richness of a language to finding similarities 
between two genome sequences.  In this paper we analyze the joint 
string complexity when both strings are generated by Markov sources. 
We prove that the joint string complexity grows linearly 
(in terms of the string lengths)
when both sources are statistically indistinguishable and 
sublinearly when sources are statistically not the same.
Precise analysis of the joint string complexity turns out to be quite challenging
requiring subtle singularity analysis and saddle point method 
over infinity many saddle points leading to novel oscillatory 
phenomena with single and double periodicities.
To overcome these challenges,
we apply powerful analytic techniques 
such as multivariate generating functions, multivariate depoissonization 
and Mellin transform, spectral matrix analysis,
and complex asymptotic methods.
\end{abstract}


\section{Introduction}

In the last decades, several attempts have been made to capture
mathematically the concept of ``complexity'' of a sequence.
The notion is connected with quite deep mathematical properties, including 
rather elusive concept of randomness in a string
(see e.g.,  \cite{IYZ02,Li93,niederreiter}), and 
the ``richness of the language''.
The \emph{string complexity} is defined
as the number of {\it distinct substrings} of the underlying string.
More precisely, if $X$ is a sequence and $I(X)$ is its set of factors
(distinct subwords), then the cardinality $|I(X)|$ is the complexity 
of the sequence.  For example, if $X=aabaa$ then 
$I(X)=\{\nu, a,b,aa,ab,ba,aab,aba,baa,aaba,abaa,aabaa\}$ 
and $|I(X)|=12$ ($\nu$ denotes the empty string).
Sometimes the complexity of a string is called the $I$-complexity~\cite{bh12}. 
This measure is simple but quite
intuitive. Sequences with low complexity contain a large number
of repeated substrings and they eventually become periodic.


In general, however, information contained in a string cannot be measured
in absolute and a reference string is required. 
To this end we introduced in~\cite{jacquet} the concept of the {\it joint} 
string complexity, or $J$-complexity, of two strings.
The $J$-complexity is the number of {\it common distinct factors} in two sequences. 
In other words, the $J$-complexity of sequences $X$ and $Y$ is equal to 
$J(X,Y)=|I(X)\cap I(Y)|$. 
We denote by $J_{n,m}$ the {\it average} value of $J(X,Y)$ when $X$ is 
of length $n$ and $Y$ is of length $m$. In this paper, we study 
the joint string complexity for Markov sources when $n=m$.  


The $J$-complexity is an efficient way of estimating similarity degree of
two strings. For example, genome sequences of two dogs will
contain more common words than genome sequences of a dog and
a cat. Similarly, two texts written in the same language  have 
more words in common than texts written in very different languages.  
Thus, the $J$-complexity is larger when languages are close ({\it e.g.} 
French and Italian), and smaller when languages are 
different ({\it e.g.} English and Polish).
In fact, texts in the same language but on different 
topics ({\it e.g.} law and cooking) have smaller $J$-complexity than 
texts on the same topic ({\it e.g.} medicine). 
Furthermore, string complexity has a variety of applications in detection
of similarity degree of two sequences, for example ``copy-paste" 
in texts or documents that will allow to detect plagiarism.
It could also be used in analysis of social networks
(e.g. tweets that are limited to 140 characters) and classification.
Therefore it could be a pertinent tool for automated monitoring
of social networks.  However, real time search in blogs, tweets and other social
media must balance quality and relevance of the content,
which -- due to short but frequent posts -- is still an unsolved problem.
However, for these short texts precise analysis is highly desirable.
We call it the "small data" problem and we hope our rigorous asymptotic analysis  
of the joint string complexity will shed some light on this problem.
In this paper we offer a precise analysis of the joint complexity  
together with some experimental results (cf. Figures~\ref{fig-eng} and \ref{fig-same})
confirming usefulness of the joint string complexity for text discrimination. 
To model real texts, we assume that both sequences are generated by Markov sources
making the analysis quite involved. To overcome these difficulties we shall
use powerful analytic techniques such as multivariate generating functions,
multivariate depoissonization and Mellin transform, spectral matrix analysis,
and saddle point methods.


String complexity was studied extensively in the past.
The literature is reviewed in \cite{jls04} where precise analysis of
string complexity is discussed for strings generated by
unbiased memoryless sources. Another analysis of the same situation
was also proposed in \cite{jacquet} where for the first time
the joint string complexity for memoryless sources was presented.
It was evident from \cite{jacquet} that precise analysis of
the joint complexity is quite challenging due to intricate
singularity analysis and infinite number of saddle points.
In this paper we deal with the joint string complexity
for Markov sources. To the best of our knowledge
this problem was never tackled before except in our recent conference paper
\cite{js12}. As expected, its
analysis is very sophisticated but at the same time quite rewarding.
It requires generalized (two-dimensional) dePoissonization and generalized
(two-dimensional) Mellin transforms.

In~\cite{jacquet} is proved that the $J$-complexity of two texts generated
by two {\it different} binary memoryless sources  grows as
$$\gamma\frac{n^\kappa}{\sqrt{\alpha\log n}}$$ 
for some $\kappa<1$ and $\gamma,\alpha>0$ depending 
on the parameters of the sources. 
When the sources are identical, then the $J$-complexity growth is $O(n)$, hence
$\kappa=1$. 
When the texts are identical (i.e, $X=Y$), then the $J$-complexity is 
identical to the $I$-complexity and it grows as $\frac{n^2}{2}$~\cite{jls04}. 
Indeed, the presence of a  common factor 
of length $O(n)$ inflates the $J$-complexity to $O(n^2)$. 

We should point out that our experiments indicate a very slow convergence
of the complexity estimates for memoryless sources. 
Furthermore, memoryless sources are not appropriate for modeling many sources, 
e.g., natural languages. In this paper and \cite{js12} 
we extend the $J$-complexity estimates 
to Markov sources of any order for a finite alphabet. Although Markov 
models are no more realistic in some applications than memoryless sources, 
they seem to be fairly good approximation for text generation. 
\begin{figure}
\centerline{\includegraphics[width=12cm]{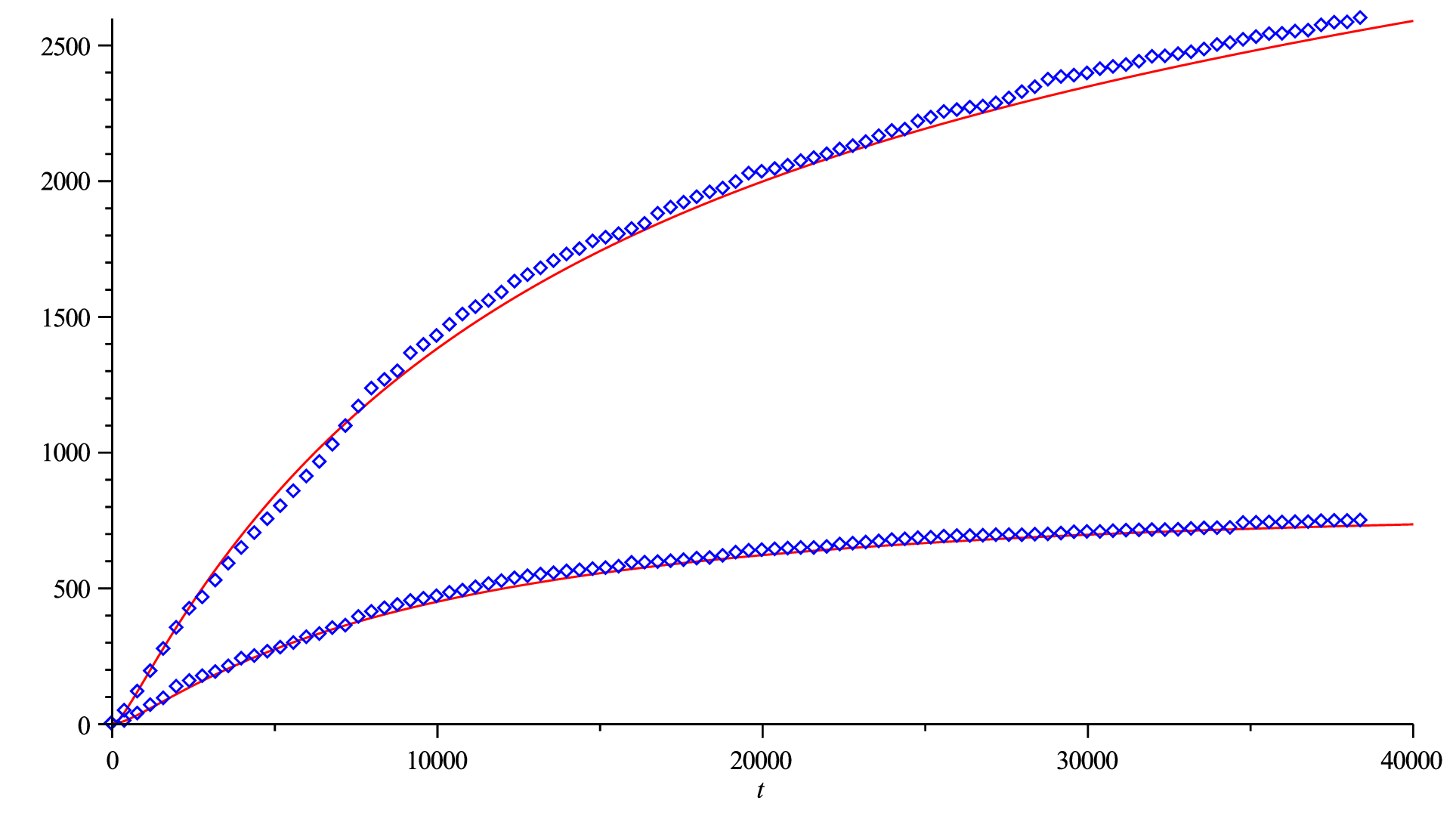}}
\caption{Joint complexity of actual simulated texts (3rd Markov order) of 
English, vs French (top), 
Polish (bottom) languages, versus average theoretical (plain).}
\label{fig-eng}
\end{figure}

Here, we derive a second order asymptotics for 
$J$-complexity for Markov sources of the following form  
$$
\gamma\frac{n^\kappa}{\sqrt{\alpha\log n+\beta}}
$$
for some $\beta>0$. This new estimate converges faster, although for small text lengths
of order $n\approx 10^2$ one needs to compute additional terms. 
In fact, for some Markov sources our analysis indicates 
that $J$-complexity oscillates with $n$.
This is manifested by appearing a periodic function  
in the leading term of our asymptotics. 
Surprisingly, this additional term even further 
improves the convergence for small values of $n$. 


Let us now summarize in full our main results Theorems~\ref{theoJC}--\ref{theogege}.
In our first main result Theorem~\ref{theoJC} we observe that the joint
string complexity $J_{n,m}$ can be asymptotically analyzed by considering a
simpler quantity called the {\it joint prefix complexity} denoted as $C_{n,m}$.  
It counts the number of common prefixes of two sets of
size $n$ and $m$, respectively,  of {\it independently} generated 
strings by two Markov sources. In the reminding part of the paper we only
deal with the joint prefix complexity $C_{n,m}$. First in 
Theorem~\ref{th-same} we considered two statistically identical sources
and prove that the joint string complexity grows linearly
with $n$: For certain sources called 
{\it noncommensurable} there is a constant in front of $n$ (that we
determine) while for {\it commensurable sources} the factor in front of $n$ is
a fluctuating periodic function of small amplitude. We shall see these two
cases permeate all our results. Then we deal 
in Theorem~\ref{th-nilpotent} with a special sources in which 
the underlying  Markov matrices
are nilpotent. After that we study general sources, however,
we split our presentation and proofs into two parts. First,
in Theorems~\ref{theogenmark0}--\ref{theo01} we 
assume that one of the source is uniform.
Under this assumption we develop techniques to prove our results.
Finally, in Theorem~\ref{th3} -- \ref{theogege} we discuss general case.

Let us now compare our theoretical results with experimental results on
real texts generated in different languages.
In Figure~\ref{fig-eng} we compare the joint complexity of a simulated
English text with the same length texts simulated in French and in Polish. 
In the simulation we use
a Markov model of order 3. It is easy to see that even for 
texts of lengths smaller than a
thousand one can discriminate between these languages. 
In fact, computations show that for English 
versus French we have $\kappa= 0.18$; 
and versus Polish: $\kappa=0.1$, 
Furthermore, for a Markov model of order 3 we find that  English text 
has entropy (per symbol): $0.944$; French: $0.934$; Polish: $0.665$.  
The theoretical curves shown in Figure~\ref{fig-eng}  are obtained through
Theorem~\ref{theogenmark}, however, for small values of the text length it
is computed via the iterative resolution of 
functional equations (\ref{jac1-bis}) and (\ref{jac2-bis}).
Figure~\ref{fig-same} shows the continuation of our theoretical estimates 
up to $n=10^{10}$ and compared with the theoretical estimate $O(n^\kappa)$
as presented in Theorems~\ref{th3} -- \ref{theogenmark}. 

\begin{figure}
\centerline{\includegraphics[width=12cm]{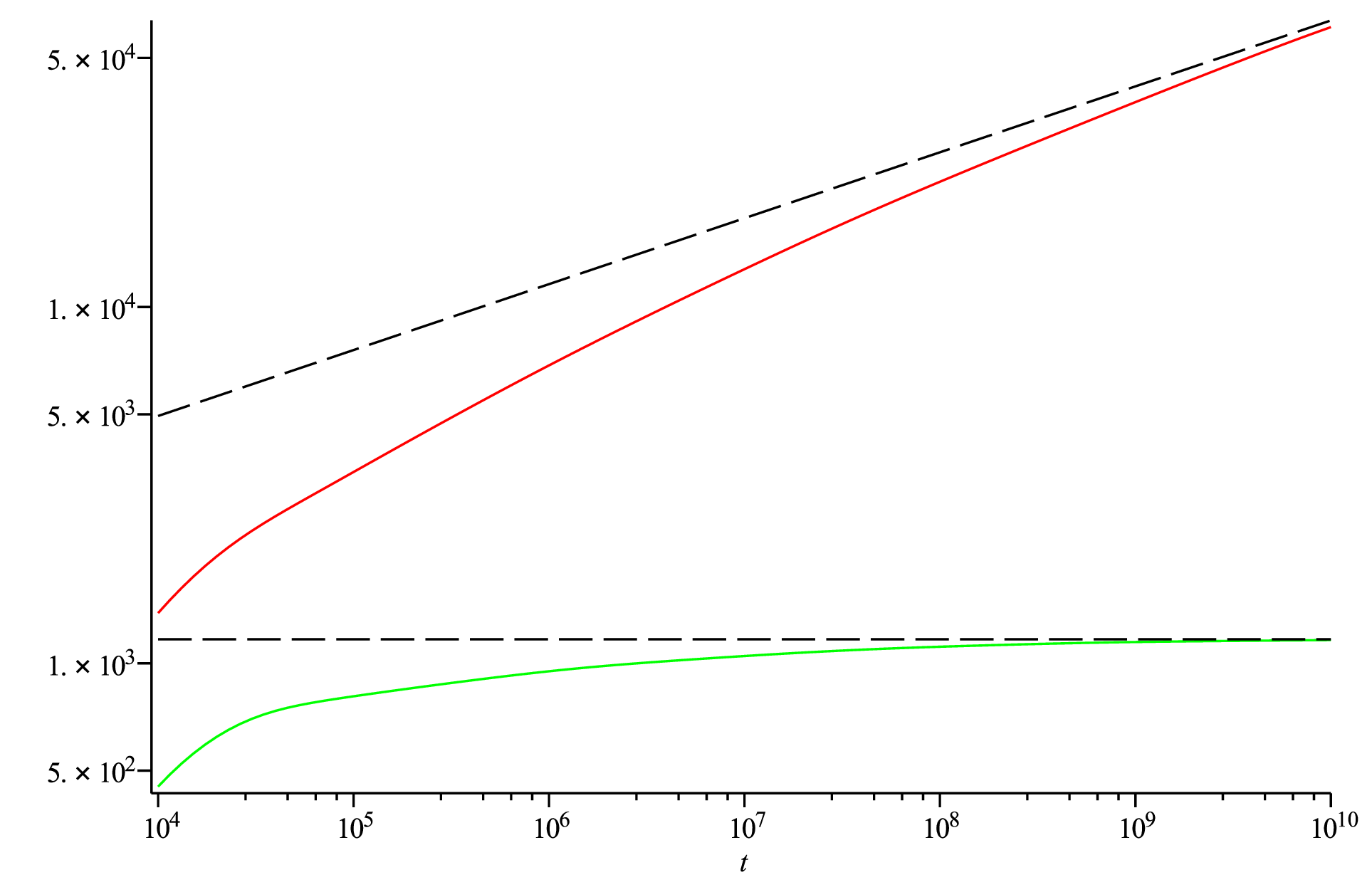}}
\caption{Average theoretical Joint complexity of 3rd Markov order text of English, vs French (top), and vs
Polish (bottom) languages, versus order estimate $O(n^\kappa)$.}
\label{fig-same}
\end{figure}

It didn't escape our attention that the joint string complexity can be
used to discriminate Markov sources \cite{ziv} since, as already observed, 
the growth of the joint string complexity is $O(n^{\kappa})$  with $
\kappa=1$ when sources are statistically indistinguishable and $\kappa<1$
otherwise. For example, we can use the joint string complexity to
verify authorship of an unknown manuscript by comparing it to a manuscript
of known authorship and checking whether $\kappa=1$ or not. 
More precisely, we propose to introduce the
following discriminant function
$$
d(X,Y)=1-\frac{1}{\log n} \log J(X,Y)
$$
for two sequences $X$ and $Y$ of length $n$. This discriminant
allows us to determine whether $X$ and $Y$ are generated by the
same Markov source or not by verifying whether
$d(X,Y) = O(1/\log n) \to 0$ or
$d(X,Y)=1-\kappa+O(\log\log n/\log n)>0$,
respectively. In fact, we used it with some success 
to classify twitter messages (see SNOW 2014 challenge of tweets detection).

The paper is organized as follows.
In the next section we present our main results Theorems~\ref{theoJC}--\ref{theogege}.
We prove Theorem~\ref{theoJC} in Section~\ref{sec-th1}. Then we present
some preliminary results in Section~\ref{sec-preliminary}. In particular, we derive
the functional equation for the joint prefix complexity $C_{n,m}$, establish some
depoissonization results, and derive double Mellin transform. 
We first prove the nilpotent case in Section~\ref{sec-nilpotent}. 
The proofs of Theorems~\ref{theogenmark0} -- \ref{theo01} are presented in
Section~\ref{sec-special}, and the proofs of Theorem~\ref{th3} -- \ref{theogege}
are discussed in Section~\ref{sec-last}.


\section{Main Results}
\label{sec-main}

In this section we define precisely our problem, introduce some
important notation, and present our main results 
Theorem~\ref{theoJC} -- Theorem~\ref{theogege}. The proofs are presented in
the remaining parts of the paper and Appendix.

\subsection{Models and notations}

We begin by introducing some general notation.
Let $X$ and $w$ be two strings over the alphabet $\CA$.
We denote by $|X|_w$ the number of times $w$ occurs in $X$
(e.g., $| abbba|_{bb}=2$ when $X=abbba$ and $w=bb$).  
By convention $|X|_\nu=|X|+1$,
where $\nu$ is the empty  string.

Throughout we denote by $X$ a string (text) whose complexity we plan to
study. We also assume that its length $|X|$ is equal to $n$.  
Then we define $I(X)=\{w: ~|X|_w\ge 1\}$, that is, $I(X)$ contains
all {\it distinct} subwords of $X$.
Observe that the string complexity $|I(X)|$ can be represented as
$$
|I(X)|=\sum_{w\in\CA^*}1_{|X|_w \ge 1}, 
$$
where $1_A$ is the indicator function of the event $A$.
Notice that $|I(X)|$ is equal to the number 
of nodes in the associated suffix tree of $X$ \cite{jls04,spa-book}
(see also \cite{js94}).

Now, let $X$ and $Y$ be two strings (not necessarily of the same length).
We define the {\it joint string complexity} as 
the cardinality of the set $J(X,Y)=I(X)\cap I(Y)$, that is,
$$
|J(X,Y)|=\sum_{w\in\CA^*}1_{|X|_w\ge 1}\times 1_{|Y|_w\ge 1}~.
$$
In other words, $J(X,Y)$ represents the number of {\it common} 
and distinct subwords of both $X$ and $Y$.
For example, if $X=aabaa$ and $Y=abbba$, then $J(X,Y)=\{\eps,a,b,ab,ba\}$.

In this paper throughout
we assume that both strings $X$ and $Y$ 
are generated by two {\it independent Markov sources} of order $r$ 
(we will only deal here with Markov of order 1, 
but extension to arbitrary order is straightforward).
We assume that source $i$, for $i\in\{1,2\}$ has the transition probabilities
$P_i(a|b)$ from state $b$ to state $a$, where $a,b\in\CA^r$. 
We denote by $\bP_1$ (resp. $\bP_2$) the transition matrix of Markov source 
1 (resp. source 2). The stationary distributions are respectively
denoted by $\vpi_1(a)$ and $\vpi_2(a)$ for $a\in \CA^r$.
Throughout, we consider general Markov sources with transition matrices $\bP_i$ 
that may contain zero  coefficients. This assumption leads to
interesting embellishment of our results. 

Let $X_n$ and $Y_m$ be two strings of respective lengths $n$ and $m$, 
generated by  Markov source 1 and Markov source 2, respectively. 
We write 
\begin{equation}
\label{eq-comp}
J_{n,m}=\E(|J(X_n,Y_m)|)-1 =
\sum_{w\in \CA^*\setminus \nu} P(|X|_w\ge 1) P(|Y|_w \ge 1)
\end{equation}
for the joint complexity, {\it i.e.} omitting the empty string. 
In this paper we study $J_{n,m}$ for $n=\Theta(m)$.

It turns out that analyzing $J_{n,m}$ is very challenging. 
It is related to the number of common nodes in two suffix trees, one
built for $X$ and the other built for $Y$. 
We know that analysis of a single suffix tree is quite challenging
\cite{js94,fw05}. Its analysis is reduced to study a simpler structure
known as {\it tries}, a digital tree built from prefixes of a set of {\it independent}
strings. We shall follow this approach here.
Therefore, we introduce another concept. Let $\CX$ be a set of infinite strings, 
and we define the prefix set $\CI(\CX)$ of $\CX$ as the set of prefixes of 
$\CX$. Let $\CX$ and $\CY$ be now two sets of strings and we 
define the {\it joint prefix complexity} as the number of common prefixes, 
{\it i.e.} $|\CI(\CX)\cap\CI(\CY)|$. When $\CX_n$ is a set of $n$ 
independent strings generated by source 1 and $\CY_m$ is a set of $m$ 
independent strings generated by source 2, then we define 
$C_{n,m}$ as 
$$
C_{n,m}=\E|\CI(\CX_n)\cap\CI(\CY_m)|)-1
$$ 
which represents the number of common prefixes between $\CX_n$ and $\CY_m$.

Observe that we can re-write $C_{n,m}$ in a different way. Define 
$\Omega^i(w)$ for $i=1,2$ as the number of strings in $\CX$ and $\CY$, respectively,
whose prefixes are equal to $w$ provided that strings in $\CX$ are generated by
source $1$ and strings in $\CY$ by source $2$. Then, it is easy to notice that
\begin{equation}
\label{eq-omega}
C_{n,m} =
\sum_{w\in \CA^*\setminus \nu} P(\Omega^1_n(w) \geq 1) P(\Omega^2_m(w) \ge 1)
\end{equation}
which should be compared to (\ref{eq-comp}).

The idea is that $C_{n,m}$ is a good approximation of $J_{n,m}$ as we present
in our first main result Theorem~\ref{theoJC}. We shall see in 
Sections~\ref{sec-preliminary} -- \ref{sec-last}
that $C_{n,m}$ are easier to analyze, however, far from simple.
In fact, $C_{n,m}$ has a nice interpretation. It corresponds to the number of
common nodes in two tries built from $\CX$ and $\CY$. We know
\cite{jst01,js-book,spa-book} that tries are easier to analyze than suffix trees.

\subsection{Summary of Main Results}

We now present our main theoretical results. 
In the first foundation result below we show that asymptotically we can analyze
$J_{n,m}$ through the quantity $C_{n,m}$ defined above in (\ref{eq-omega}).
The proof of the next result can be found in Section~\ref{sec-th1}.

\begin{theorem}\it
Let $n$ and $m$ be of the same order.
Then there exists $1/2\le\eps<1$ such that  
\begin{equation}
\label{eq-th1}
J_{n,m}=C_{n,m}+O(n^\eps+m^\eps)
\end{equation}
as $n\to\infty$.
\label{theoJC}
\end{theorem}

In the rest of the paper we shall analyze $C_{n,n}$. We should point out that
the error term could be as large as the leading term, but for  sources 
that are relatively close the error term will be negligible. 

Now we presents a series of results each
treating different cases of Markov sources. 
However, our results depend on whether the underlying Markov sources 
are commensurable or not so we define them next.

\begin{definition}[Rationally Related Matrix]
We say that a matrix $\bM=[m_{ab}]_{(a,b)\in\CA^2}$ is
{\it rationally related} if $\forall (a,b,c)\in\CA^3$ we have
$
m_{ab}+m_{ca}-m_{cb}\in\mathbb{Z},
$
where $\mathbb{Z}$ is the set of integers.
\end{definition}

\begin{definition}[Logarithmically rationally related matrix]
We say that a matrix $\bM=[m_{ab}]$ is {\it logarithmically
rationally related} if there exists a non zero real number $x$ such that 
the matrix $x\log^*(\bM)$ is rationally related, where the matrix 
$\log^*(\bM)$ is composed of 
$\log(m_{ab})$ when $m_{ab}>0$ and zero otherwise.
The smallest non negative value $\omega$ of the real $x$ defined above 
is called the root of $\bM$.
\end{definition}

The following matrix is an example of logarithmically rationally related matrix:
$$
\bP=\left[
\begin{array}{cccc}
\frac{1}{4}&\frac{1}{8}&\frac{1}{2}&\frac{1}{4}\\
\frac{1}{4}&\frac{1}{8}&\frac{1}{8}&\frac{1}{2}\\
\frac{1}{4}&\frac{1}{4}&\frac{1}{4}&\frac{1}{8}\\
\frac{1}{4}&\frac{1}{2}&\frac{1}{8}&\frac{1}{8}
\end{array}
\right] .
$$
Its root is $1/log 2$.

\begin{definition}[Logarithmically commensurable pair]
We say that a pair of two matrices $\bM=[m_{ab}]_{(a,b)\in\CA^2}$ and
$\bM'=[m'_{ab}]$ is {\it logarithmically commensurable} if there exist
a  pair of real numbers $(x,y)$ such that
$x\log^*(\bM)+y\log^*(\bM')$ is not null and is logarithmically rationally related.
\end{definition}

Notice that when $\bM$ and $\bM'$ are both rationally related, then  the pair is
logarithmically commensurable. Nevertheless it is possible to have
logarithmically commensurable pairs with the individual matrices not
logarithmically rationally related. For example when
$\log^*\bM'=2\pi\bQ+\log^*\bM$ with $\bQ$ an integer matrix.

We are now in the position to discuss our first main result 
for Markov sources that are statistically indistinguishable.
Throughout we  present results for $m=n$.

\begin{theorem}
\label{th-same}\it
Consider the average joint complexity of two texts of length $n$ 
generated by the same  general stationary Markov source, that is,
$\bP:=\bP_1=\bP_2$.
\parsec
{\rm (i)} [{\sl Noncommensurable Case}.]
Assume that $\bP$ is not logarithmically rationally related.
Then
\begin{equation}
\label{p1}
J_{n,n}=\frac{2n\log 2}{h}+o(n)
\end{equation}
where $h$ is the entropy rate of the source defined as
$h=\sum_{a,b \in \CA} \pi(a) P(a|b)$. 
\parsec
{\rm (ii)} [{\sl Commensurable Case}.]
Assume that $\bP$ is logarithmically rationally related.
Then there is $\eps<1$ such that: 
\begin{equation}
\label{p2}
J_{n,n}=\frac{2n\log 2}{h}(1+Q_0(\log n))+O(n^{\eps})
\end{equation}
where $Q_0(.)$ is a periodic function of small amplitude.
(In Section~\ref{sec-same} compute explicitly $Q$.)
\end{theorem}

Now we consider sources that are not the same and have respective 
transition matrices $\bP_1$ and $\bP_2$. The transition matrices 
are on $\CA^r\times\CA^r$.
If $(a,b)\in \CA^r\times\CA^r$, we denote by $P_i(a|b)$ the 
$(a,b)$-th coefficient of matrix $\bP_i$. For a tuple of 
complex numbers $(s_1,s_2)$ we write $\bP(s_1,s_2)$ for the following  matrix 
$$
\bP(s_1,s_2)=[\bP_1(a|b))^{-s_1}(\bP_2(a|b))^{-s_2}]_{a,b\in A}.
$$ 
In fact, we can write it as the Schur product, denoted as $\star$, of two matrices 
$\bP_1(s_1)=\bP_1(a|b))^{-s_1}$ and $\bP_2(s_2)=(\bP_2(a|b))^{-s_2}$, that is,
$\bP(s_1,s_2)=\bP_1(s_1) \star \bP_2(s_2)$.

To present succinctly our general results we need some more notation.
Let $\la {\bf x} | {\bf y} \ra$ be the scalar product of vector ${\bf x}$ and
vector ${\bf y}$. By $\lb(s_1,s_2)$ we denote the main eigenvalue of matrix 
$\bP(s_1,s_2)$, and $\bu(s_1,s_2)$ its corresponding right eigenvector 
(i.e, $\lb(s_1,s_2)\bu(s_1,s_2)=\bP(s_1,s_2)\bu(s_1,s_2)$), and
$\bzeta(s_1,s_2)$ its left eigenvector (i.e., $\lb(s_1,s_2)\bu(s_1,s_2)=
\bu(s_1,s_2)\bP(s_1,s_2)$). We assume that $\la\bzeta(s_1,s_2)|\bu(s_1,s_2)\ra=1$.
Furthermore, the vector $\bpi(s_1,s_2)$ is defined as the vector 
$(\pi_1(a)^{-s_1}\pi_2(a)^{-s_2})_{a\in\CA}$ where 
$(\pi_i(a))_{a\in\CA}$ is the left eigenvector of matrix 
$\bP_i$ for $i\in\{1,2\}$. In other words $(\pi_i(a))_{a\in\CA}$ 
is the stationary distribution of the Markov source $i$.

We start our presentation with the simplest case, namely
the case when the matrix $\bP(0,0)$  is nilpotent \cite{hj}, that is, 
for some $K$ the matrix $\bP^K(0,0)$ is the null matrix. 
Notice that for nilpotent matrices $\forall(s_1,s_2)$:~ $\bP^K(s_1,s_2)=0$. 

\begin{theorem}
\label{th-nilpotent}\it
If $\bP(s_1,s_2)$ is nilpotent, then there exists $\gamma_0$ 
such that 
\begin{equation}
\label{p3}
\lim_{n\to\infty}J_{n,n}=\gamma_0:=\la\bone_C(\bI-\bP(0,0))^{-1}|\bone\ra
\end{equation}
where
$\bone$ is the unit vector, $\bone_C$ the vector on $\CA$ with $\bone_C(a)=1$ 
when $a$ is common to both sources, and $\bone_C(a)=0$ otherwise. 
\end{theorem}

This result is not surprising and rather trivial since the common 
factors can only occur in a finite window at the beginning of the strings. 
It turns out that $\gamma_0 =1168$ for 3rd order Markov model of 
English versus Polish languages used in our experiments.

Throughout, now we assume that $\bP(s_1,s_2)$ is not nilpotent. 
We need to pay much closer attention to 
the structure of the set of roots of 
the {\it characteristic equation}
$$\lambda(s_1,s_2)=1$$
that will play a major role in the analysis.
We discuss in depth properties of these roots in 
Section~\ref{sectCK}. Here we introduce only a few important definitions.

\begin{definition}
The kernel $\overline{\CK}$ is the set of complex tuples $(s_1,s_2)$ 
such that $\bP(s_1,s_2)$ has its largest eigenvalue equal to 1.
The real kernel $\CK=\overline{\CK}\cap\mathbb{R}^2$, {\it i.e.} the set of 
real tuples $(s_1,s_2)$ such that the main eigenvalue $\lambda(s_1,s_2)=1$. 
\end{definition}

The following lemma is easy to prove.

\begin{lemma}
The real kernel forms a concave curve in $\mathbb{R}^2$.
\label{lemconc}\end{lemma}

Furthermore, we introduce two important notations:
\begin{eqnarray}
\kappa&=&\min_{(s_1,s_2)\in\CK}\{-s_1-s_2\} \label{p4a}\\
(c_1,c_2)&=&\arg\min_{(s_1,s_2)\in\CK}\{-s_1-s_2\}.
\label{p4b}
\end{eqnarray}

This leads to a new concept $\partial\CK$, the border of the kernel
$\overline{\CK}$, defined as follows.

\begin{definition}
We denote $\partial\CK$ the subset of $\overline{\CK}$ made of the pairs 
$(s_1,s_2)$ such $\Re(s_1,s_2)=(c_1,c_2)$. 
\end{definition}

Easy algebra shows that $\kappa \leq 1$. Furthermore, in the Appendix
we prove the following property.

\begin{lemma}
Let $c_1$ and $c_2$ minimize $-s_1-s_2$ where real tuple $(s_1,s_2)\in\CK$. 
Assume $\forall(a,b)\in\CA^2$: $P_1(a|b)>0$ then $c_1\le 0$ and $c_2\ge -1$. 
\label{lemc2}
\end{lemma}

The case when both matrices $\bP_1$ and $\bP_2$ have some
zero coefficients is the most intricate part. Therefore, 
to present our strongest results, we start with a special case when one of 
the source is uniform. Later we generalize it.

We first consider a special case 
when source 1 is uniform memoryless, {\it i.e.} 
$\bP_1=\frac{1}{|\CA|}\bone\otimes\bone$ and the other matrix 
$\bP_2$ is not nilpotent and general (that is, it may have some
zero coefficients). In this case we always have 
$c_1<0$ and $c_2<0$. 
This case we have the following theorem. 


\begin{theorem}
\label{theogenmark0}
Let $\bP_1=\frac{1}{|\CA|}\bone\otimes\bone$ and $\bP_2\neq\bP_1$ is a general
transition matrix. Thus both $c_1$ and $c_2$ are between $-1$ and $0$.
\parsec
{\rm (i)}  [{\sl Mono periodic case}.]
If $\bP_2$ is not logarithmically rationally related,
then there exists a periodic function $Q_1(x)$ of small amplitude such that 
\begin{equation}
\label{p7.5}
C_{n,n}=\frac{\gamma_2 n^\kappa}{\sqrt{\alpha_2\log n+\beta_2}}
(1+Q_1(\log n) +o(1)).
\end{equation}
\parsec
{\rm (ii)} [{\sl Double periodic case}.]
If $\bP_2$ is 
logarithmically rationally related,
then there exists a double periodic function\footnote{
We recall that a double periodic function is a function on real numbers 
that is a sum of two periodic functions of non commensurable periods.}
$Q_2(.)$ of small amplitude such that 
\begin{equation}
\label{p8}
C_{n,n}=\frac{\gamma_2 n^\kappa}{\sqrt{\alpha_2\log n+\beta_2}}
(1+Q_2(\log n) +o(1)).
\end{equation}
\end{theorem}

The constants  $\gamma_2$, $\alpha_2$ and $\beta_2$ in the Theorem~\ref{theogenmark0} 
are explicitly computable as presented next. 
To simplify our notation for all $(a,b)\in\CA^2$ we shall write
$P_2(a|b)=P(a|b)$ and $\bP_2=\bP$. Therefore
\be
\bP(s_1,s_2)=|\CA|^{s_1}\bP(s)
\ee
with $\bP(s)=\bP(0,s)$.
We also write $\pi(a)=\pi_2(a)$ and
$\bpi(s)=[\pi(a)^{-s}]_{a\in \A}$, thus
\be
\bpi(s_1,s_2)=|\CA|^{s_1}\bpi(s_2)~.
\ee

Let $\lb(s_1,s_2)$ be again the main (largest) eigenvalue of $\bP(s_1,s_2)$. We have
\be
\lambda(s_1,s_2))=|\CA|^{s_1}\lambda(s_2)
\ee
where $\lb(s)$ is the main eigenvalue of matrix $\bP(s)$. We also define
$\bun(s)$ as the right eigenvector of $\bP(s)$ and
$\bzeta(s)$ as the left eigenvector
provided $\la\bzeta(s)|\bun(s)\ra=1$. It is easy to see that
$$
\lb(s)=\la\bzeta(s)|\bP(s)\bun(s)\ra.
$$

Now we can express $c_1$ and $c_2$ defined in (\ref{p4b}) in another way,
Notice that if $\lb(s_1,s_2)=1$, then in this case 
$$
s_1=- \log_{|\CA|} \lambda(s_2).
$$
Define $L(s)= \log_{|\CA|} \lambda(s)$. Then
$c_2$ is the value that minimizes $L(s)-s$, that is,
\be
\label{sw3}
\frac{\lb'(c_2)}{\lb(c_2)}=\log|\CA|.
\ee
Also $c_1=-L(c_2)$ and $\kappa=-c_1-c_2=\min_s\{L(s)-s\}$.
We have $\kappa\le1$ since $L(0)=1$.

We now can  presents explicit expression for the
constants in Theorem~\ref{theogenmark0}.

\begin{theorem}
We consider the case $\bP_1=\frac{1}{|\CA|}\bone\otimes\bone$ and $\bP_2\neq\bP_1$ has all non negative coefficients. 
Let $f(s)=\la\bpi(s)|\bun(s)\ra$ and
$g(s)=\la\bzeta(s)|\bone\ra$. Furthermore, with $\Psi(s)$ being the Euler
psi function, define  $\alpha_2=L''(c_2)$ where $L(s) = \log_{|\CA|} \lambda(s)$, and
$$
\beta_2(s_1,s_2)=-\alpha_2\left(\Psi(s_1)+\frac{1}{1+s_1}+\log|\CA|\right) 
$$
$$
+\Psi'(s_1)-\frac{1}{(s_1+1)^2}+\Psi'(s_2)-\frac{1}{(s_2+1)^2} 
$$
$$
+\frac{f''(s_2)}{f(s_2)}-\left(\frac{f'(s_2)}{f(s_2)}\right)^2+\frac{g''(s_2)}
{g(s_2)}-\left(\frac{g'(s_2)}{g(s_2)}\right)^2
$$
as well as
$$
\gamma(s_1,s_2)=\frac{f(s_1)g(s_2)(s_1+1)\Gamma(s_1)(s_2+1)
\Gamma(s_2)}{\lb(s_2)\log|\CA|\sqrt{2\pi}}.
$$
We have $c_1=-\log_{|\CA|}\lb(c_2)$, and 
then
\begin{eqnarray}
\nonumber
C_{n,n}&=&n^\kappa\frac{\gamma(c_1,c_2)}{\sqrt{\alpha_2 \log n+\beta_2(c_1,c_2)}} \\
&&+n^\kappa Q(\log n)+o(\frac{n^\kappa}{\sqrt{\log n}}). 
\label{sw4}
\end{eqnarray}
The function $Q(x)$ can be expressed as
$$
Q(x)=\sum_{(s_1,s_2)\in\partial\CK^*}e^{ix\Im(s_1+s_2)} 
\frac{\gamma(s_1,s_2)}{\sqrt{\alpha_2 x+ \beta_2(s_1,s_2)}}.
$$
If the matrix $\bP_2$ is logarithmically rationally related, 
then $\partial\CK$ is a lattice.
Let $\omega$ the root of $\bP_2$
then
$$
\partial\CK=\left\{\left(c_1+\frac{2ik\pi}{\log|\CA|},
c_2+2i\pi\ell\omega\right), k,\ell)\in\mathbb{Z}^2\right\},
$$
and $\sqrt{x}Q(x)$ is asymptotically double periodic.
Otherwise (i.e., irrational case),
$$
\partial\CK=\left\{\left(c_1+\frac{2ik\pi}{\log 2},c_2\right),~ 
k\in\mathbb{Z}\right\}
$$
and $\sqrt{x}Q(x)$ is asymptotically single periodic.
The amplitude of $Q$ is of order $10^{-6}$.
\label{theo01}
\end{theorem}

Now we consider the case when the matrices $\bP_1$ and $\bP_2$ are general 
and $\bP_1\neq\bP_2$. If they contain some zero coefficients, then
we may have $\bP(-1,0)\neq\bP_1$ and/or $\bP(0,-1)\neq\bP_2$. 
In this (very unlikely) case we may have $\bP(-1,0)=\bP(0,-1)$ or more
generally they are conjugate.
For example for matrices:
$$
\bP_1=\left[
\begin{array}{cccc}
0&0&\frac{5}{8}&\frac{1}{2}\\
0&\frac{1}{8}&\frac{1}{8}&\frac{1}{4}\\
\frac{1}{2}&\frac{1}{4}&\frac{1}{4}&\frac{1}{8}\\
\frac{1}{2}&\frac{5}{8}&0&\frac{1}{8}
\end{array}
\right]
\bP_2=\left[
\begin{array}{cccc}
\frac{1}{2}&\frac{5}{8}&0&\frac{1}{2}\\
\frac{1}{2}&\frac{1}{8}&\frac{1}{8}&\frac{1}{4}\\
0&\frac{1}{4}&\frac{1}{4}&\frac{1}{8}\\
0&0&\frac{5}{8}&\frac{1}{8}
\end{array}
\right]
$$
we have
$$
\bP(-1,0)=\bP(0,-1)=\left[
\begin{array}{cccc}
0&0&0&\frac{1}{2}\\
0&\frac{1}{8}&\frac{1}{8}&\frac{1}{4}\\
0&\frac{1}{4}&\frac{1}{4}&\frac{1}{8}\\
0&0&0&\frac{1}{8}
\end{array}
\right]
$$
If $\bP(-1,0)=\bP(0,-1)$, then there is no unique solution $(c_1,c_2)$ 
of the characteristic equation, 
and therefore there is no saddle point. 
As a consequence, we have the following theorem:

\begin{theorem}
When $\bP(-1,0)$ and $\bP(0,-1)$ are conjugate matrices we have
\be
C_{n,n}=\gamma_0(-\kappa)n^\kappa(1+o(1))
\ee
where $\kappa<1$ is such that $(-\kappa,0)\in\CK$ and 
$\gamma_0(-\kappa)$ are explicitly computable.
When both matrices are logarithmically rationally related, then
\be
C_{n,n}=\gamma(-\kappa)n^\kappa\left(1+Q_0(\log n)+O(n^{-\eps}\right)
\ee
where $Q_0(.)$ is small periodic function and $\eps>0$.
\label{theo11}
\end{theorem}

In general, however, $\bP(-1,0)$ and $\bP(0,-1)$ are {\it not} conjugate, and
therefore, there is a unique $(c_1,c_2)$ of the characteristic equation
$\lambda(c_1,c_2)=1$. 
As in the special case discussed above, $c_1>-1$ and $c_2>-1$, however,
our results are quantitatively different when
$c_1>0$ or $c_2>0$. We consider it first in Theorem~\ref{th3} below.
Since both cases cannot occur simultaneously, 
we dwell only on the case $c_2>0$; the case $c_1>0$ 
can be handled in a similar manner.

\begin{theorem}
\label{th3}\it
Assume $\bP(0,0)$ is not nilpotent and $c_2>0$.
\parsec
{\rm (i)} [{\sl Noncommensurable Case}.]
We assume that $\bP(-1,0)$ is not logarithmically related. Let $-1<c_0<0$ such 
that $(c_0,0)\in\CK$. 
There exist $\gamma_1$ such that
\begin{equation}
\label{p5}
C_{n,n}=\gamma_1 n^{-c_0}(1+o(1))
\end{equation}
\parsec
{\rm (ii)} [{\sl Commensurable Case}.]
Let now $\bP(-1,0)$ be logarithmically rationally related.
There exists a periodic function $Q_1(.)$ of small amplitude such that  
\begin{equation}
\label{p61}
C_{n,n}=\gamma_1 n^{-c_0}(1+Q_1(\log n)+O(n^{-\eps}).
\end{equation}
\end{theorem}

Finally, we handle the most intricate case when
 both $c_1$ and $c_2$ are between $-1$ and $0$.  
Recall that when both matrix $\bP_1$ and $\bP_2$ have all positive
coefficients, then  $\bP(-1,0)=\bP_1$ and $\bP(0,-1)=\bP_2$.

\begin{theorem}
\label{theogenmark}\it
Assume that both $c_1$ and $c_2$ are between $-1$ and $0$ and 
$\bP(0,0)$ is not nilpotent.
\parsec
{\rm (i)} [{\sl Non periodic Case}.]
If $\bP(-1,0)$ and $\bP(0,-1)$ are not logarithmically commensurable matrices, 
then there exist 
$\alpha_2$, $\beta_2$ and $\gamma_2$ such that
\begin{equation}
\label{p7}
C_{n,n}=\frac{\gamma_2 n^\kappa}{\sqrt{\alpha_2\log n+\beta_2}}(1+o(1)).
\end{equation}
\parsec
{\rm (ii)}  [{\sl Mono periodic case}.]
If only one of the matrices $\bP(-1,0)$ and $\bP(0,-1)$ is logarithmically rationally related,
then there exists a periodic function $Q_2(x)$ of small amplitude such that 
\begin{equation}
\label{p7.51}
C_{n,n}=\frac{\gamma_2 n^\kappa}{\sqrt{\alpha_2\log n+\beta_2}}
(1+Q_1(\log n) +o(1)).
\end{equation}
\parsec
{\rm (iii)} [{\sl Double periodic case}.]
If both matrices $\bP_1$ and $\bP_2$ are 
logarithmically rationally related,
then function $Q_2(.)$ is double periodic with 
small amplitude such that 
\begin{equation}
\label{p81}
C_{n,n}=\frac{\gamma_2 n^\kappa}{\sqrt{\alpha_2\log n+\beta_2}}
(1+Q_2(\log n) +o(1)).
\end{equation}
In the three cases the constants  $\gamma_2$, $\alpha_2$ and $\beta_2$ are explicitly computable. 
\end{theorem}
\parsec
{\bf Remark}
In the rational case, the root can be such that 
$\omega\log|\CA|$ is rational. In this case $\sqrt{x}Q(x)$ tends to a simple 
periodic function instead of a double periodic function. 

At last, we provide explicit expression for some of constants in previously
stated results, in particular in the most interesting Theorem~\ref{theogenmark}. We denote $H(s_1,s_2)=1-\lb(s_1,s_2)$. Let $H_1(s_1,s_2)=\frac{\partial}{\partial s_1}\lb(s_1,s_2)$, $f(s_1,s_2)=\la\bpi(s_1,s_2)|\bun(s_1,s_2)\ra$ and $g(s_1,s_2)=\la\bzeta(s_1,s_2)|\bone\ra$.

\begin{theorem}
Let $c_1$ and $c_2$ be between $-1$ and $0$.
In the general case we have
\be
C_{n,n}=n^\kappa\sum_{(s_1,s_2)\in\partial\CK}\frac{\gamma(s_1,s_2)n^{-i\Im(s_1+s_2)}}{\sqrt{\alpha_2\log n+\beta(s_1,s_2)}}\left(1+O(\frac{1}{\log n})\right).
\ee
With $H_1$, $\alpha_2$,  $\beta_2(s_1,s_2)$ and $\gamma(s_1,s_2)$ such
\begin{eqnarray*}
H_1&=&\frac{\partial}{\partial s_1}H(c_1,c_2)\\
\alpha_2&=&\frac{\Delta H-2\frac{\partial^2}{\partial s_1\partial s_2}H}{H_1}|_{(s_1,s_2)=(c_1,c_2)}\\
\beta_2(s_1,s_2)&=&-\frac{\alpha_2}{2}\left(\Psi(s_1)+\Psi(s_2)+\frac{\fps1f +\fps2 f}{f}+\frac{\fps1 g+\fps2 g}{g}\right)\\
&&+\Psi'(s_1)+\Psi'(s_2)+\frac{\Delta f-2\frac{\partial^2}{\partial s_1\partial s_2}f}{f}+\frac{\Delta g -2\frac{\partial^2}{\partial s_1\partial s_2}g}{g}\\
&&-\left(\frac{\fps1 f-\fps2 f}{f}\right)^2-\left(\frac{\fps1 g-\fps2 g}{g}\right)^2\\
&&+\frac{\alpha_2^2}{2}-\frac{\frac{\partial^3}{\partial s_1^3}H-\frac{\partial^3}{\partial s_1^2\partial s_2}H-\frac{\partial^3}{\partial s_1\partial s_2^2}H+\frac{\partial^3}{\partial s_2^3}H}{2H_1}
+\left(\frac{\frac{\partial^2}{\partial s_1^2}H-\frac{\partial^2}{\partial s_2^2}H}{2H_1}\right)^2
\end{eqnarray*}
and
$$
\gamma(s_1,s_2)=\frac{f(s_1)g(s_2)(s_1+1)\Gamma(s_1)(s_2+1)
\Gamma(s_2)}{\lb(s_2)\log|\CA|\sqrt{2\pi}}.
$$

\label{theogege}\end{theorem}

The expression for $\alpha_2$ seems to be asymmetric in $(s_1,s_2)$. 
In fact, it is not since the maximum of $s_1+s_2$ for $(s_1,s_2)\in\CK$ 
attained on $(c_1,c_2)$ necessarily implies that 
$\frac{\partial}{\partial s_1}H(c_1,c_2)=\frac{\partial}{\partial s_2}H(c_1,c_2)$.
Formally the constant $\alpha_2$ is equal to $\Delta\lb_1(c_1,c_2)
\langle\bone|\nabla\lb_1(c_1,c_2)\rangle$ where $\nabla$ is the 
gradient operator, $\Delta$ the Laplacian operator $\frac{\partial^2}{\partial 
s_1^2}+\frac{\partial^2}{\partial s_2^2}$.

\def\gutter{2cm}

\begin{figure}
\centerline{\includegraphics[width=12cm]{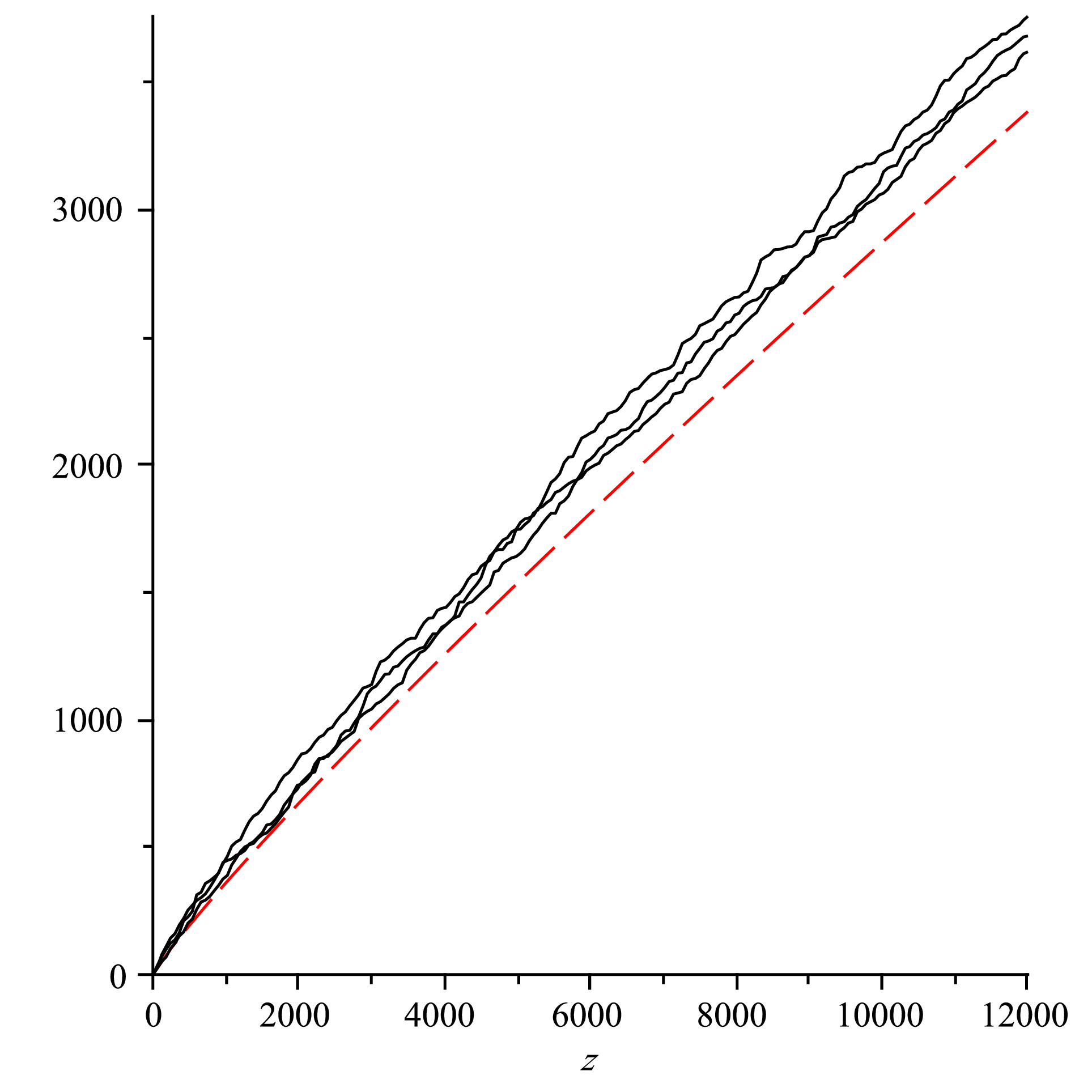}}
\caption{Joint complexity: three simulated trajectories (black)
versus asymptotic average (dashed red) for the case $c_2>0$.}
\label{fig1a}
\end{figure}

\begin{figure}
\centerline{\includegraphics[width=12cm]{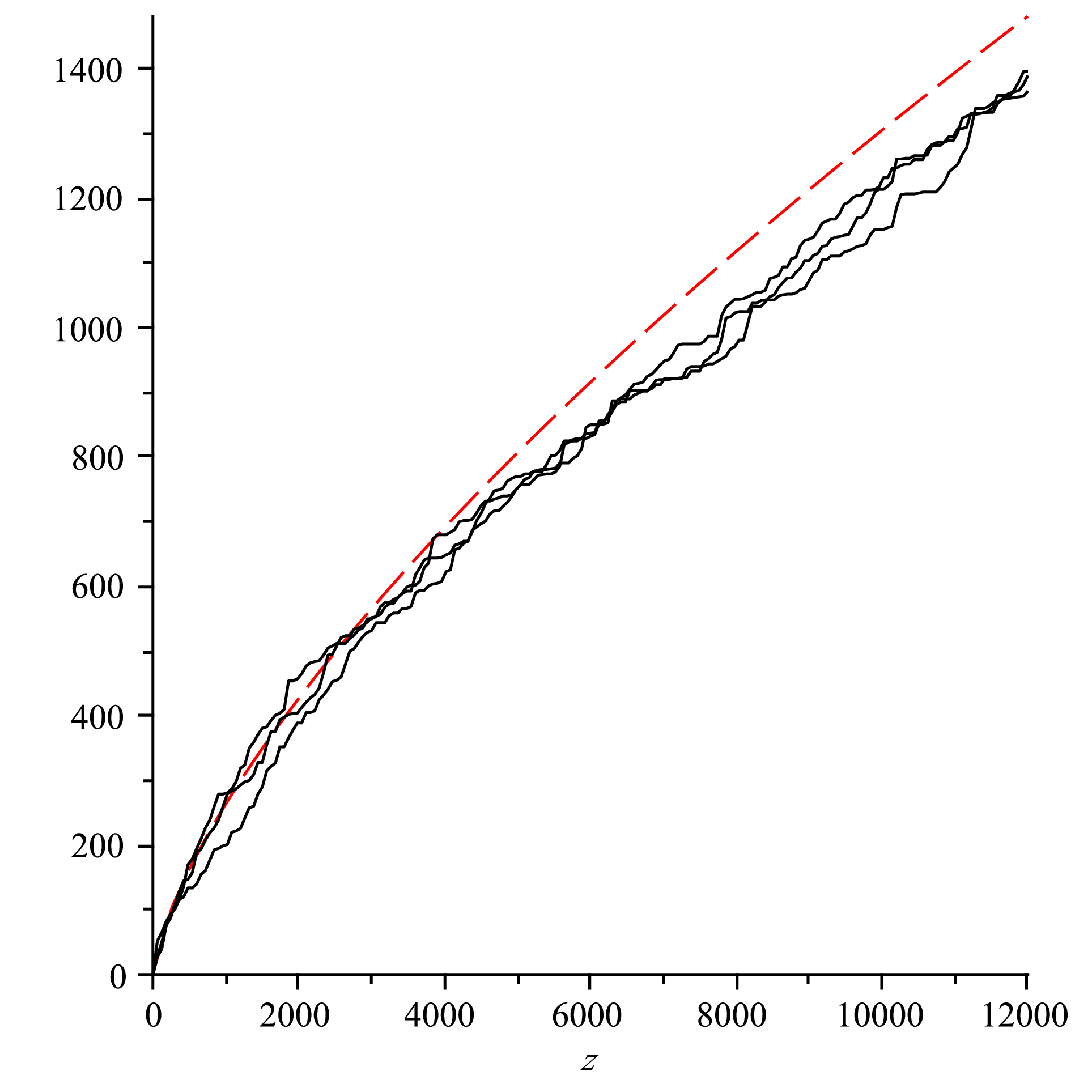}}
\caption{Joint semi-complexity: three simulated trajectories (black)
versus asymptotic average (dashed red)for the case $c_2<0$.}
\label{fig1b}
\end{figure}

Finally, we illustrate our results on two examples.
\med
{\bf Example}.
In Figures~\ref{fig1a} and \ref{fig1b} we plot the joint complexity for several
pairs of strings $X$ and $Y$.
String $X$ is generated by a Markov source with
the transition matrix $\bP$, and string $Y$ is generated by
a uniform memoryless source.
We consider two Markov sources for $X$ with
the following transition matrix:

\be
\bP=\left[\begin{array}{cc}
 0&0.5\\
 1&0.5
 \end{array}\right]
\hspace{\gutter}
\bP=\left[\begin{array}{cc}
 0.2&0.8\\
 0.8&0.2
 \end{array}\right] .
\ee
For the first $\bP$ in Figure~\ref{fig1a} we have $c_2>0$ (see Theorem~\ref{th3})
while for the second $\bP$ in Figure~\ref{fig1b} we have $c_2<0$
(cf. Theorem~\ref{theogenmark0}(i)).

\section{Proof of Theorem~\ref{theoJC}}
\label{sec-th1}

We recall that $X$ and $Y$ are two independent strings of length $n$ and $m$,
generated by two Markov sources characterized by transition
matrices $\bP_1$ and $\bP_2$, respectively. In the previous
section we write $|X|_w$ and $|Y|_w$ for the number of $w\in \CA$ occurrences
in $X$ and $Y$. But it will be convenient to use another notation for
these quantities, namely
$$
O^1_n(w):= |X|_w, \ \ \ \ O^2_m(w):=|Y|_w.
$$
We shall use this notation interchangeably.
Finally, we write $\A^+=\A^*-\{\nu\}$, that is, 
for the set of all nonempty words.
As observed in (\ref{eq-comp}) we have
\be
J_{n,m}=\sum_{w\in\CA^+}P(O_n^1(w)\ge 1) P(O_m^2(w) \ge 1).
\ee

In \cite{rs98,js-book} the generating function of $P(O_n(w)\ge 1)$ for
a Markov source is derived. It involves the 
{\it autocorrelation} polynomial of $w$, as discussed below. 
However, to make our analysis tractable we notice that $O_n^i(w)\ge 1$,
$i=1,2$, is equivalent to $w$ being a prefix of at least one of 
the $n$ suffixes of $X$. But this is not sufficient to push forward
our analysis. We need a second much deeper observation that replaces
{\it dependent suffixes} with {\it independent strings} 
to shift analysis from suffix trees to tries, as 
already observed in \cite{js94} and briefly discussed in 
Section~\ref{sec-main}.
In order to accomplish it, we consider two {\it sets} of 
strings of size $n$ and $m$, respectively, generated 
{\it independently} by Markov sources $\bP_1$ and $\bP_2$.
As in Section~\ref{sec-main} we denote by $\Omega_n^i(w)$ the number 
of strings for which $w$ is a prefix
when there are $n$ strings generated by source
$i$, for $i\in\{1,2\}$. The average {\it joint prefix complexity} 
satisfies (\ref{eq-omega}) that we repeat below
\begin{equation}
\label{chap10-purdue1}
C_{n,m}=\sum_{w\in\CA^+} 
P(\Omega^{1}_n(w)\ge 1)P(\Omega^2_m(w)\ge1).
\end{equation}


Before we prove our first main result Theorem~\ref{theoJC} 
we need some preliminary work.
First, observe that it is relatively easy to compute the probability
$P(\Omega^i_n(w)\ge 1)$. Indeed,
$$
P(\Omega^i_n(w)\ge 1)=1-(1-P_i(w))^n.
$$
Notice that the quantity $1-(1-P(w))^n$ is the probability 
that $|\CX_n|_w>0$ where $\CX_n$ is a set of $n$ independently generated 
strings.  To prove Theorem~\ref{theoJC} we must show that when $n\to\infty$
$$
P(O_n^i(w) \ge 1) \sim P(\Omega^i_n(w)\ge 1)
$$
which we do in the next key lemma.

We denote by $\CB_k$  the set of words of length $k$ such that a word 
$w\in \CB_k$ does not overlap with itself over more than $k/2$ characters 
(see~\cite{js-book,js94,fw05} for more precise definition).
It is proven in~\cite{fw05} that 
$$
\sum_{w\in\CA^k-\CB_k}P(w)=O(\delta_1^k)
$$ 
where $\delta_1<1$ is the largest element of the Markovian transition 
matrix $\bP$. In order to allow some transition probabilities to be equal to $0$  
we define 
\begin{eqnarray*}
p&=&\exp\left(\lim\sup_{k,w\in\CA^k}\frac{\log P(w)}{k}\right)\\
q&=& \exp\left(\lim\inf_{k,w\in\CA^k,P(w)\neq 0}\frac{\log P(w)}{k}\right). 
\end{eqnarray*}
These quantities exist and are smaller than 1 since $\CA$ is a 
finite alphabet \cite{pittel85,spa-book}. In fact, they are related to Renyi's entropy
of order $\pm \infty$, respectively \cite{spa-book}. 
We write $\delta=\sqrt{p}<1.$

Let $X_n$ be a string of length $n$ generated by a Markov source. 
For $w \in \CA^*$, define  
\begin{equation}
\label{eq-d}
d_n(w)=P(O_n(w)>0)-(1-(1-P(w))^n).
\end{equation}
We prove the following lemma

\begin{lemma}\it
Let $w\in\CA^k$ be of length $k$. There exists $\rho>1$ and a 
sequence $R_n(w)=O(P(w)\rho^{-n})$ 
such that for all $1>\eps>0$ we have:
\parsec
{\rm (i)} for $w\in\CB_k$: $d_n(w)=O((nP(w))^\eps k\delta^k)+R_n(w)$;
\parsec
{\rm(ii)} for $w\in\CA^k-\CB^k$: $d_n(w)=O((nP(w))^\eps\delta^k)+R_n(w)$.
\end{lemma}

\begin{proof}
Let $N_0(z)=\sum_{n\ge 0}P(O_n(w)=0)z^n$. We know from~\cite{js-book} that 
$$N_0(z)=\frac{S_w(z)}{D_w(z)}$$ where 
$S_w(z)$ is the autocorrelation polynomial of word $w$
and $D_w(z)$ is defined as follows
\be
\label{eq-D}
D_w(z)=S_w(z)(1-z)+z^kP(w)\left(1+F_{w_1,w_k}(z)(1-z)\right),
\ee
where $k=|w|$ is the length of word $w$ with first symbol $w_1$  
and the last symbol $w_k$. Here  $F_{a,b}(z)$ for $(a,b)\in\CA^2$ is 
a generating function that depends on the Markov sources, 
as describe below. We also write $F_w(z)$ when $w_1=a$ and $w_k=b$. 

Let $\bP$ be the transition matrix of the Markov source.  Let $\bpi$ be its 
stationary vector and for $a\in\CA$ let $\pi_a$ be its coefficient at symbol $a$. 
The vector $\bone$ is the vector with all coefficients equal to 1 and $\bI$ 
is the identity matrix. Assuming that the symbol $a$ (resp. $b$) is the first 
(resp. last) character of $w$, we have \cite{rs98}
\be
F_{w}(z):=F_{a,b(z)}=\frac{1}{\pi_a}\left[(\bP-\bpi\otimes\bone)
\left(\bI-z(\bP+\bpi\otimes\bone)\right)^{-1}\right]_{a,b}
\ee
where $[A]_{a,b}$ denotes the $(a,b)$th coefficient of matrix $A$. 
An alternative way to express $F_w(z)$ is
\be
F_{w}(z)=\frac{1}{\pi_a}\langle\bbe_a (\bP-\bpi\otimes\bone)
\left(\bI-z(\bP+\bpi\otimes\bone)\right)^{-1}\bbe_b\rangle
\ee
where $\bbe_c$ for $c\in\CA$ is the vector with a 1 at the position corresponding
to symbol $c$ and all other coefficients are 0.

By the spectral representation of matrix $\bP$ we have \cite{hj}
$$
\bP=\bpi\otimes\bone+\sum_{i>1}\lambda_i\bu_i\otimes\bzeta_i
$$ 
where $\lambda_i$ for $i\ge 1$ is the $i$th eigenvalue 
(in decreasing order) of matrix $\bP$ (with $\lambda_1=1$),  
and $\bu_i$ (resp. $\bzeta_i$) are their corresponding right (resp. left) 
eigenvectors. Thus
\be
(\bP-\bpi\otimes\bone)\left(\bI-z(\bP+\bpi\otimes\bone)\right)^{-1}=
\sum_{i>1}\frac{\lambda_i}{1-\lambda_i z}\bu_i\otimes\bzeta_i
\ee
and therefore the function $F_{w}(z)$ is defined for all $z$ such that 
$|z|<\frac{1}{|\lambda_2|}$ and is uniformly $O(\frac{1}{1-|\lambda_2 z|})$.

We now follow the approach from~\cite{fw05} that extends to Markovian sources 
the analysis presented in~\cite{js94} for memoryless sources 
(see also \cite{js-book}).
Let
$$\Delta_w(z)=\sum_{n\ge 0}d_n(w)z^n$$
be the generating function of 
$d_n(w)$ defined in (\ref{eq-d}). After some algebra we arrive at
\begin{eqnarray}
\Delta_w(z)&=&\frac{P(w)z}{1-z}\left(\frac{1+(1-z)F_{w}(z)}{D_w(z)} 
-\frac{1}{1-z+P(w)z}\right).
\end{eqnarray}
We have 
$$d_n(w)=\frac{1}{2i\pi}\oint \Delta_w(z)\frac{dz}{z^{n+1}},$$ 
integrated on any loop encircling the origin in the definition 
domain of $d_w(z)$. Extending the result from~\cite{js94}, 
the authors of~\cite{fw05} show that there exists $\rho>1$ 
such that the function $D_w(z)$ defined in (\ref{eq-D}) has a single root in the disk of 
radius $\rho$. Let $A_w$ be this root. We have via the residue formula
\be
d_n(w)=\Res(\Delta_w(z),A_w)A_w^{-n}-(1-P(w))^n+d_n(w,\rho)
\ee
where $\Res(f(z),A)$ denotes the residue of function $f(z)$ on 
complex number $A$. Thus
\be
d_n(w,\rho)=\frac{1}{2i\pi}\oint_{|z|=\rho}\Delta_w(z)\frac{dz}{z^{n+1}}.
\ee

We have
\be
\Res(\Delta_w(z),A_w)=\frac{P(w)\left(1+(1-A_w)F_{w}(A_w)\right)}{(1-A_w)C_w}
\ee
where $C_w=D'_w(A_w)$. But since $D_w(A_w)=0$ we can write
\be
\Res(\Delta_w(z),A_w)=-\frac{A_w^{-k}S_w(A_w)}{C_w}.
\ee

We can take the asymptotic expansion of $A_w$ and $C_w$ as it is 
described in~\cite{js-book}, in Lemma 8.1.8 and Theorem 8.2.2. 
Anyhow the expansions were done in the
memoryless case. But an extension to Markov sources simply consists in replacing $S_
w(1)$ into $S_w(1)+P(w)F_w(1)$, so we find
\be
\left\{
\begin{array}{rcl}
A_w&=&1+\frac{P(w)}{S_w(1)}\\
&&+P^2(w)\left(\frac{k-F_w(1)}{S^2_w(1)}-\frac{S_w'(1)}{S^3_w(1)} \right)+O(P(w)^3),\\
C_w&=&-S_w(1)+P(w)\left(k-F_w(1)-2\frac{S_w'(1)}{S_w(1)}\right)\\
&&+O(P(w)^2).
\end{array}
\right.
\ee
These expansions also appear in~\cite{fw05}.

>From now the proof takes the same path as the proof of Theorem  8.2.2 
in~\cite{js-book}. We define the function
\be
d_w(x)=\frac{A_w^{-k}S_w(A_w)}{C_w}A_w^{-x}-(1-P(w))^{x}.
\ee
More precisely, we define the function $\td_w(x)=d_w(x)-d_w(0)e^{-x}$. 
Its Mellin transform \cite{spa-book} is 
$$
d^*_w(s)\Gamma(s)=\int_0^\infty\td_w(x)x^{s-1}dx
$$ 
defined for all $\Re(s)\in(-1,0)$ with
\begin{eqnarray}
d_w^*(s)&=&\frac{A_w^{-k}S_w(A_w}{C_w}((\log A_w)^{-s}-1)\nonumber\\
&&+1-\left(-\log(1-P(w)))\right)^{-s}
\end{eqnarray}
where $\Gamma(s)$ is the Euler gamma function.
When $w\in\CB_k$ with the expansion of $A_w$ and since $S_w(1)=1+O(\delta^k)$ and
$S'_w(1)=O(k\delta^k)$, we find that similarly as in~\cite{js-book}
\be
d^*_w(s)=O(|s|k\delta^k)P(w)^{1-s}
\ee
and therefore by the reverse Mellin transform, for all $1>\eps>0$:
\begin{eqnarray}
\td_w(n)&=&\frac{1}{2i\pi}\int_{-\eps-i\infty}^{-\eps+i\infty}d^*_w(s)\Gamma(s)
s^{-n}ds\nonumber\\
&=&O(n^{\eps}P(w)^{\eps}k\delta^k).
\end{eqnarray}
When $w\in\CA^k-\CB_k$ it is not true that $S_w(1)=1+O(\delta^k)$, thus it is shown 
in~\cite{fw05} that there exists $\alpha>0$ such that for all 
$w\in\CA^*$: $S_w(z)> \alpha$ for all $z$ such that $|z|\le\rho$. 
Therefore we get $\td_w(n)=O(n^{\eps}P(w)^{\eps})$.

We set
\be
R_n(w)=d_w(0)e^{-n}+d_n(w,\rho).
\ee
We first investigate the quantity $d_w(0)$. We prove that $d_w(0)=O(P(w))$. 
Noticing that 
$$ S_w(A_w)=S_w(1)+\frac{P(w)}{S_w(1)}S_w'(1) +O(P(w)^2)$$
we have the expansion
\be
-\frac{A_w^{-k}S_w(A_w)}{C_w}=1-\frac{P(w)}{S_w(1)}\left(F_w(1)+\frac{S_w'(1)}
{S_w(1)}\right)+O(P(w)^2).
\ee
Thus
\be
d_w(0)=-\frac{P(w)}{S_w(1)}\left(F_w(1)+\frac{S_w'(1)}{S_w(1)}\right)+O(P(w))^2).
\ee
Thus $d_w(0)=O(P(w))$.

Now we need to consider $d_n(w,\rho)$. 
Since $\Delta_n(z)$ is clearly $O(P(w))$ and the integral 
$\oint\Delta(z)\frac{dz}{z^{n+1}}$ 
is over the circle of radius, the result is $O(P(w)\rho^{-n})$.
\end{proof}

Now we are ready to prove Theorem~\ref{theoJC}.
\begin{proof}
Again staring with
\be
J_{n,m}=\sum_{w\in\CA^*}P(O^1_n(w)>0)P(O^2_m(w)>0)
\ee
we note that
\begin{eqnarray*}
P(O^1_n(w)>0))&=&1-(1-P_1(w))^n+d^1_n(w),\\
P(O^2_m(w)>0)&=&1-(1-P_2(w))^m+d^2_m(w).
\end{eqnarray*}
Thus
\begin{eqnarray}
J_{n,m}&=&C_{n,m}+\sum_{w\in\CA}d_n^1(w) P(O^2_m>0)\nonumber\\
&+&\sum_{w\in\CA}(1-(1-P_1(w))^n)d_m^2(w).
\end{eqnarray}
We will develop the proof for the first sum, since the proof for the second proof being
somewhat similar. When $w\in\CB_k$, we have for all $\eps>0$
\be
d^1_n(w)=O(n^{\eps} P_1(w)^{\eps}k\delta_1^k)+R^1_n(w)
\ee
and the $R_n^1(w)$ terms are all $O(P_1(w)\rho^{-n})$. We look at the sum 
$$\sum_k\sum_{w\in\CB_k}n^\eps P(w)^\eps k\delta_1^k.$$ 
It is smaller than 
$$\sum_k\sum_{w\in\CA^k}n^\eps P_1(w)k(q^{\eps-1}\delta_1)^k$$ 
which is equal to $n^\eps\sum_k k(q^{\eps-1}\delta_1)^k$. By choosing a value 
$\eps_1$ of $\eps$ enough close to $1$ so that $q^{\eps_1-1}\delta_1<1$ 
we have the $n^{\eps_1}$ order. Notice that with $\delta_1= \sqrt{p}$ we must
conclude that $1/2< \eps_1<1$.

When $w\in\CA^k-\CB_k$ the $\delta_1^k$ factor disappears in the
right-hand side expression for $d^1_n(w)$. But in this case 
$$
\sum_{w\in\CA^k-\CB_k}n^\eps P(w)(q^{\eps-1})^k =
O(n^\eps\delta_1^k(q^{\eps-1})^k,
$$
and we conclude similarly.

It remains the sum $\sum_{w\in\CA^*}R_n^1(w)P_2(O_m(w)>0)$. For this we remark that 
$P_2(O_m(w)>0)=O(mP_2(w))$. Therefore the sum is of order 
$\rho^{-n}m\sum_{w\in\CA^*}P_1(w)P_2(w)$. It turns out that 
$$\sum_{w\in\CA^k}P_1(w)P_2(w)=O(\lambda_{12} ^k)$$ 
where $\lambda_{12}$ is the main eigenvalue of the Shur product matrix 
$\bP_1\star\bP_2$ (also known denoted $\bP(-1,-1)$. 
Since $\lambda_{12}<1$ the sum converges and is $O(\rho^{-n}m)$.
\end{proof}

\section{Some Preliminary Results}
\label{sec-preliminary}

In this section we first derive the recurrence on $C_{n,m}$ which will lead to
the functional equation on the double Poisson transform $C(z_1,z_2)$ of $C_{n,m}$, 
that in turn allows us to  find the
double Mellin transform $C^*(s_1,s_2)$ of $C(z_1,z_2)$. 
Finally applying a double depoissonization
we first recover the original $C_{n,n}$ and ultimately the joint 
string complexity $J_{n,n}$ through Theorem~\ref{theoJC}.

\subsection{Functional Equations}

Let $a\in\CA$ and define 
$$
C_{a,m,n}=\sum_{w\in a\CA*}P(\Omega^1_n(w)\ge 1 )P(\Omega^2_m(w)\ge 1)
$$ 
where $w\in a\CA^*$ means that $w$ starts with an $a\in \A$.
We recall that $\Omega^i_n(w)$ represents the number of strings of length $n$
that start with prefix $w$.

Notice that $C_{a,m,n}=0$ when $n=0$ or $m=0$.
Using Markov nature of the string generation, 
the quantity $C_{a,n,m}$ for $n,m\ge 1$ satisfies the following 
recurrence for all $a,b\in\CA$
\begin{eqnarray*}
C_{b,n,m}&=&1+\sum_{a\in\CA}\sum_{n_a,m_a}{n\choose n_a}
{m\choose m_a}\\
&&\times(P_1(a|b))^{n_a}(1-P_1(a|b))^{n-n_a}\\
&&\times (P_2(a|b))^{m_a}
(1-P_2(a|b))^{m-m_a}C_{a,n_a,m_a},
\end{eqnarray*}
where $n_a$ (resp. $m_a$) denotes the number of strings among 
$n$ (resp. $m$) independent strings 
from source $1$ (resp. $2$) that have symbol $a$ followed by symbol $b$. 
Indeed, partitioning $b\A^*$ as $\{b\} +\sum_{a\in \A} a\A^*$ we obtain
the recurrence noting that strings are independent and 
${n\choose n_a} (P_1(a|b))^{n_a}(1-P_1(a|b))^{n-n_a}$
is the probability of $n_a$ out of $n$ starting with $ba$. 
starts 

In a similar fashion, the {\it unconditional} average $C_{n,m}$ 
satisfies for $n,m\geq 2$
$$
C_{n,m}=1+\sum_{a\in\CA}\sum_{n_a,m_a}{n\choose n_a}
{m\choose m_a} \pi_1^{n_a}(a)(1-\pi_1(a))^{n-n_a}
$$
$$
\times
\pi_2^{m_a}(a)(1-\pi_2(a))^{m-m_a} C_{a,n,m}.
$$

To solve it we introduce the double Poisson transform of $C_{a,n,m}$ as
\be
C_a(z_1,z_2)=\sum_{n,m\geq 0}C_{a,n,m}\frac{z_1^nz_2^m}{n!m!}e^{-z_1-z_2}
\ee
that translates the above recurrence into the following functional equation:
\beq
C_b(z_1,z_2)&=&(1-e^{-z_1})(1-e^{-z_2})\nonumber\\
&&+\sum_{a\in\CA}C_a
\left(P_1(a|b) z_1,P_2(a|b)z_2\right).\label{jac1-bis}
\eeq
To simplify it, we define the double Poisson transform
\be
C(z_1,z_2)=\sum_{n,m\geq 0}C_{n,m}\frac{z_1^nz_2^m}{n!m!}e^{-z_1-z_2}
\ee
finding that
\beq
C(z_1,z_2)&=&(1-e^{-z_1})(1-e^{-z_2})\nonumber\\
&+&\sum_{a\in\CA}C_a(\pi_1(a)
z_1,\pi_2(a)z_2).
\label{jac2-bis}
\eeq

Our goal now is to find asymptotic expansion of 
$C(z_1, z_2)$ as $z_1,z_2 \to \infty$ in a cone around the real axis. 
This will be accomplished in the next subsection
using double Mellin transform. Granted it, we shall appeal 
to a double depoissonization result to recover asymptotically $C_{n,m}$ and
ultimately $J_{n,m}$. 

\subsection{Double DePoissonization}
\label{sec-depo}

Once we know $C(z_1, z_2)$ for $z_1=n, z_2=m \to \infty$ we then need to recover
$C_{n,m}$. Double depoissonization lemma discussed and proved in \cite{js-book}
(see Lemma 10.3.4) allows us to do exactly that but in order to apply it we 
need to  postulate some conditions on the underlying Poisson transforms.
We briefly review double depoissonization next.

For a double sequence $a_{n,m}$  define
\begin{eqnarray*}
f(z_1,z_2)&=&\sum_{m,n=0}^\infty a_{n,m}\frac{z_1^n}{n!} e^{-z_1} 
\frac{z_2^m}{m!}e^{-z_2}, \\
f_n(z_2)&=&\sum_{m=0}^\infty a_{n,m}\frac{z_2^m}{m!}e^{-z_2}.
\end{eqnarray*}
We notice that $f(z_1,z_2)$ is the Poisson transform of the sequence
$f_n(z_2)$ with respect to the variable $z_1$.
Now we postulate certain conditions on $f(z_1,z_2)$ and $f_n(z_2)$ that will
allow us to extract asymptotics of $a_{n,m}$ from $f(z_1,z_2)$.
\med
{\sf First depoissonization}.
For $z_2\in\S_\th:=\{z_2:~ \arg(z_2)<\th\}$ 
we postulate that there exist constants $\beta, \alpha, B$ and $D$
such that
\begin{eqnarray*}
z_1\in\S_\th: |f(z_1,z_2)|&<&B(|z_1|^\beta+|z_2|^\beta)\\
z_1\notin\S_\th: |f(z_1,z_2)e^{z_1}|&<&D|z_2|^\beta e^{\alpha |z_1|}.
\end{eqnarray*}
Therefore, from the one-dimensional analytic depoissonization of
\cite{js98,spa-book} for $z_2\in\S_\th$, we have for all integers $k>0$
$$
f_n(z_2)=f(n,z_2)+O\left(n^{\beta-1}+\frac{|z_2|^\beta}{n}\right)+
O(|z_2|^\beta n^{\beta-k}).
$$
Similarly, when $z_2\notin\S_\th$ we postulate
\begin{eqnarray*}
z_1\in\S_\th: |f(z_1,z_2)e^{z_2}|&<&D|z_1|^\beta e^{\alpha |z_2|}\\
z_1\notin\S_\th: |f(z_1,z_2)e^{z_1+z_2}|&<&De^{\alpha |z_1|+\alpha|z_2|}.
\end{eqnarray*}
Thus for all integer $k$ and $\forall z_2\notin\S_th$
$$
f_n(z_2)e^{z_2}=f(n,z_2)e^{z_2}+O(n^{\beta-1}e^{\alpha|z_2|})+O(n^{\beta-k}
e^{\alpha|z_2|}).
$$
\parsec
{\sf Second depoissonization}.
The above two conditions on $f_n(z_2)$, respectively for $z_2\in\S_\theta$
and $z_2\notin\S_\th$,
allow us to depoissonize $f_n(z_2)$. For all $k>\beta$:
\begin{itemize}
\item for $z_2\in\S_\th$: $f_n(z_2)=O(n^{\beta}+|z_2|^\beta)$;
\item for $z_2\notin\S_\th$: $f_n(z_2)e^{z_2}=O(n^\beta e^{\alpha|z_2|})$.
\end{itemize}
These estimates are uniform. Therefore,
$$
a_{n,m}=f_n(m)+O\left(\frac{n^{\beta}}{m}+\frac{m^\beta}{n}\right)+
O\left(n^\beta m^{\beta-k}\right).
$$
Since
$$
f_n(m)=f(n,m)+O\left(n^{\beta-1}+\frac{m^\beta}{n}\right)
$$
and setting $k>\beta+1$, we find the desired estimate.

Now we are ready for formulate our depoissonization lemma.
In \cite{js-book} it is shown that $C(z_1,z_2)$ satisfies depoissonization conditions
for memoryless sources. In Appendix we prove the following lemma.

\begin{lemma}[DePoissonization]\it
We have
$$
C_{n,m}=C(n,m)+O\left(\frac{n}{m} + \frac{m}{n}\right)
$$ 
for large $n$ and $m$.
\label{lemDePo}
\end{lemma}

To find $C(n,m)$ from $C(z_1,z_2)$ we follow the Mellin transform approach,
however for general sources we need to consider a double Mellin transform.
We start in the next subsection with a simple case when $\bP_1=\bP_2$.


\subsection{Mellin Transform for $\bP_1=\bP_2$: Proof of Theorem~\ref{th-same}}
\label{sec-same}

We first present a general result for identical Markov sources, that is,
$\bP_1=\bP_2=\bP$ proving  Theorem~\ref{th-same}. In this case 
(\ref{jac1-bis}) can be rewritten with $c_a(z)=C_a(z,z)$:
\be
c_b(z)=(1-e^{-z})^2+\sum_{a\in\CA}c_a\left(P(a|b) z\right).
\ee
This equation is directly solvable by the 
Mellin transform  defined as
$$
c_a^*(s)=\int_0^\infty c_a(x)x^{s-1}dx
$$ 
that exists in the fundamental strip $-2<\Re(s)<-1$. 
Properties of Mellin transform can be found in
\cite{fs-book,spa-book}. It follows that
for all $b\in\CA$ \cite{spa-book} 
\be
\label{ws1}
c_b^*(s)=(2^{-s}-2)\Gamma(s)+\sum_{a\in\CA}(P(a|b))^{-s}c_a^*(s).
\ee
It is better to write it in the matrix form. Let $\bc(s)=[c^*_a(s)]_{a\in \A}$
be the vector of Mellin transforms $c^*_a(s)$ and
and $\bP(s)=[P^{-s}(a|b)]_{a,b\in A}$. Then (\ref{ws1}) becomes
$$
\bc(s)=(2^{-s}-2)\bone(\bI-\bP(s))^{-1}
$$
where, again, $\bone$ is the vector of dimension $|\CA|$ made of all $1$'s,
and $\bI$ is the identity matrix.

We now can derive the Mellin transform
$c^*(s)$ of $c(z):=C(z,z)$ representing the unconditional joint string 
complexity. From above and (\ref{jac2-bis}) we arrive at
$$
c^*(s)=(2^{-s}-2)\Gamma(s)+\sum_{a\in\CA}(\pi(a))^{-s}c_a^*(s)~.
$$
which in matrix form can be rewritten as
\be
\label{ws2}
c^*(s)=(2^{-s}-2)\Gamma(s)\left(1+\langle\bone(\bI-\bP(s))^{-1}|\vpi(s)
\rangle\right)
\ee
where $\vpi(s)$ is the the vector made of coefficients 
$\pi(a)^{-s}$ and we recall $\langle \bx |\by \rangle$ is the
inner product of vectors $\bx$ and $\by$. 

To find the behavior of $c(z)$ for large $z$ near the real axis we apply
the inverse Mellin approach as discussed in \cite{fs-book,spa-book}.
We observe that 
$$
c(z) = \frac{1}{2\pi i} \int_{-3/2 -\infty}^{-3/2+\infty} c^*(s) z^{-s} ds. 
$$
The asymptotics of $c(z)$ for $|\arg(z)|<\theta$ is given 
by the residues of the 
function $c^*(s) z^{-s}$ occurring at the poles $s=-1$ and $s=0$.  
They are respectively equal to 
$$
\frac{2\log 2}{h}z
$$
and 
$$
-1-\langle\bone(\bI-\bP(0,0))^{-1}\vpi(0)\rangle. 
$$
The first residues comes from the 
singularity of $(\bI-\bP(s))^{-1}$ at $s=-1$. This leads to 
Theorem~\ref{th-same}(i). When $\bP$ is rationally related 
then there are additional poles on a countable set of complex numbers 
$s_k$ regularly spaced on the line $\Re(s_k)=-1$, 
and such that $\bP(s_k)$ has eigenvalue 1. These poles contributes 
to the periodic terms of Theorem~\ref{th-same}(ii).  The proof
of Theorem~\ref{th-same} is now complete.


\subsection{Double Mellin transform: Case $\bP_1 \neq \bP_2$}

 From now on we only consider the case $\bP_1\neq\bP_2$, and therefore we
need to study properties of $C_a(z,z)$ through double Mellin transform defined as
\be
C_a^*(s_1,s_2)=\int_0^\infty\int_0^\infty 
C_a(z_1,z_2)z_1^{s_1-1}z_2^{s_2-1}dz_1dz_2
\label{eqMellin}
\ee
or similarly the Mellin transform $\bC^*(s_1,s_2)$ applied  to
$\bC(z_1,z_2)$, provided we can find a strip $a< \Re(s_1), \Re(s_2) < b$ 
where the above transforms exist.
Since for any $y\in\mathbb{R}^+$ and a function $f$ we have the identity 
$$\int_0^\infty f(yx)x^{s-1}dx=a^{-s}\int_0^\infty f(x)x^{s-1}dx$$
we conclude that
\be
\bC^*(s_1,s_2)=\Gamma(s_1)\Gamma(s_2)\bone+\bP(s_1,s_2)\bC^*(s_1,s_2)
\ee
or formally
\be
\label{purdue1}
\bC^*(s_1,s_2)=\Gamma(s_1)\Gamma(s_2)\left(\bI-\bP(s_1,s_2)\right)^{-1}\bone.
\ee

However, the above formal derivation needs to be amended with a careful analysis of
the convergence issues, which we do next. Notice that for any $a\in\CA$: 
$C_a(z_1,z_2)=O(|z_1|+|z_2|)$ when $z_1,z_2\to\infty$.  
But as easy to check $C_a(z_1,z_2)$ is also $O(|z_1|+|z_2|)$ when $z_1,z_2\to 0$, 
therefore the Mellin transform is not appropriately defined in~(\ref{eqMellin}). 
To correct it, we now introduce correction terms in the expression of 
$C_a(z_1,z_2)$ so that the corresponding Mellin transform exists 
for $-2<\Re(s_1),\Re(s_2)<-1$.  

To continue, we now define a slightly modified Mellin transform, namely 
$$
\tbC(z_1,z_2)=\bC(z_1,z_2)-\bD(z_1,z_2)
$$ 
where 
$$
\bD(z_1,z_2)=z_1 e^{-z_1}\bC_1(z_2)+z_2 e^{-z_2}\bC_2(z_1)-\bC_{1,1}z_1 z_2 e^{-z_1-z_2}
$$
with 
\begin{eqnarray*}
\bC_1(z)&=&\frac{\partial}{\partial z_1}\bC(z_1,z_2)|_{(z_1,z_2)=(0,z)}\\
\bC_2(z)&=&\frac{\partial}{\partial z_2}\bC(z_1,z_2)|_{(z_1,z_2)=(z,0)}.
\end{eqnarray*}
Notice that $\tbC(z_1,z_2)$ is now $O(|z_1|^2+|z_2|^2)$ when 
$z_1,z_2\to 0$. We can show that 
$\bC(z_1,z_2)=O(|z_1|+|z_2|)$ for $(z_1,z_2)$ in four dimension 
cone containing $\mathbb{R^+}\times\mathbb{R}^+$, 
therefore by Ascoli theorem $\frac{\partial}{\partial z_1}\bC(z_1,z_2)
=O(1)$ in the same cone, $D_1(z)$ is $O(|z|)$ for $z\in\mathbb{R}^+$, 
and similarly  $D_2(z)$ is $O(1)$. All of this is to state that 
$\tbC(z_1,z_2)$ is $O(|z_1|+|z_2|)$ when $z_1,z_2\to\infty$, thus
the Mellin transform of $\tbC(z_1,z_2)$ is well defined for 
$\Re(s_1),\Re(s_2)\in(-2,-1)$. Let $\tbC^*(s_1,s_2)$ be
the corresponding Mellin transform. 
 
 For $a\in\CA$ let ${C_1}_a(z)$ be the coefficient of the vector 
$\bC_1(z)$ corresponding to the symbol $a$. 
For $b\in\CA$ we have the functional equation
 \be
 {C_1}_b(z)=1-e^{-z}+\sum_{a\in\CA}P_1(a|b){C_1}_a(P_2(a|b)z)
 \ee 
 and the Mellin transform of $\bC_1(z)$, say $\bC_1^*(s)$ formally satisfies 
 \be
 \bC_1(s)=-\Gamma(s)\bone+\bP(-1,s)\bC_1(s)
 \ee
 or 
 \be
 \bC_1(s)=-\Gamma(s)(\bI-\bP(-1,s))^{-1}\bone .
 \ee
Similarly the Mellin transform $\bC_2(s)$ of $\bC_2(z)$ 
satisfies $\bC_2(s)=\Gamma(s)(\bI-\bP(s,-1))^{-1}\bone$. To finish, 
we notice that $\bC_{1,1}=(\bI-\bP(-1,-1))^{-1}$.

 Denoting $\hbC(s_1,s_2)=(\bI-\bP(s_1,s_2))^{-1}\bone$, we find  
 \begin{eqnarray*}
 \tbC^*(s_1,s_2)&=&\Gamma(s_1)\Gamma(s_2)\left(\hbC(s_1,s_2)+s_2\hbC(s_1,-1)\right.\\
 &&\left.+s_1\hbC(-1,s_2)+s_1s_2\hbC(-1,-1)\right)
 \end{eqnarray*}
finally leading to 
\be
\label{eth1}
C^*(s_1,s_2)=\Gamma(s_1)\Gamma(s_2)
\left(1+\langle\vpi(s_1,s_2)|\hbC(s_1,s_2)\rangle\right)
\ee
where $\vpi(s_1,s_2)$ denotes
the vector composed of $\pi_1(a)^{-s_1}\pi_2(a)^{-s_2}$ for 
$a\in\CA$ and $\langle|\rangle$ is the vector internal product. 

Our goal is to find $C_{n,n}$ (i.e., $n=m$). But by
depoissonization it is asymptotically equal to $C(n,n)$, therefore we must find $\bC(z,z)$
which by the inverse Mellin transform becomes
\begin{eqnarray*}
\bC(z,z)&=&\bD(z,z)+
\frac{1}{(2i\pi)^2}\int_{\rho_1}
\int_{\rho_2}\tbC^*(s_1,s_2)z^{-s_1-s_2}ds_1ds_2.
\end{eqnarray*}
After some algebra we finally arrive at
\begin{eqnarray}
C(z,z)&=&(1-e^{-z})^2+\frac{1}{(2i\pi)^2}\int_{\rho_1}
\int_{\rho_2}\Gamma(s_1)\Gamma(s_2)\nonumber\\
&&\times\langle\vpi(s_1,s_2)|\hbC(s_1,s_2)+s_2\hbC(s_1,-1)\nonumber\\
&&s_1\hbC(-1,s_2)+s_1s_2\hbC(-1,-1)\rangle\nonumber\\
&&z^{-s_1-s_2}ds_1ds_2
\label{eq-dml}
\end{eqnarray}
where the integration is over the lines $\Re(s_1)=\rho_1$ and
$\Re(s_1)=\rho_2$ with
$(\rho_1,\rho_2)$ belonging to the fundamental strip of $C^*(s_1,s_2)$: $(-2,-1)$. 
We shall analyze asymptotically (\ref{eq-dml}) in the next sections.


\subsection{Properties of the Kernel}
\label{sectCK}


We recall from Section~\ref{sec-main} that we define the kernel 
$\overline{\CK}$ as the set of complex tuples $(s_1,s_2)$ such that $\bP(s_1,s_2)$ has
largest eigenvalue $\lambda(s_1,s_2)=1$. Furthermore, we
also define $\partial\CK$ as the subset of $\overline{\CK}$ consisting 
of the pairs $(s_1,s_2)$ such $\Re(s_1,s_2)=(c_1,c_2)$
where
$$
(c_1,c_2)=\arg\min_{(s_1,s_2)\in\CK}\{-s_1-s_2\}.
$$
We also denote $\partial\CK^*=\partial\CK-\{(c_1,c_2)\}$. 

Let us start with the structure of the set $\partial\CK$.

\begin{definition}\it
Let $\bP$ be a matrix on $\CA\times\CA$ of complex coefficients $p_{ab}$ 
for all $(a,b)\in\CA^2$. Let $\bQ$ be a matrix $q_{ab}$. In the following 
we say $\bP$ and $\bQ$ are {\it conjugate} if there exists a non-zero complex vector 
$(x_a)_{a\in\CA}$ such that $q_{ab}=\frac{x_a}{x_b}p_{ab}$. 
We say that such matrices are {\it imaginary conjugate} 
if $|x_a|=1$ for all $a\in \A$.
\label{def-conjugate}
\end{definition}

Observe that: (i) two conjugate matrices have the same eigenvalue set;
(ii) if $\bun=(u_a)_{a\in\CA}$ is right eigenvector of $\bP$,
then $(x_au_a)_{a\in\CA}$ is right eigenvector of $\bQ$. 
Similarly, if $(\zeta_a)_{a\in\CA}$ is left eigenvector of 
$\bP$, then $(\frac{1}{x_a}\zeta_a)_{a\in\CA}$ is the left eigenvector of $\bQ$. 

The following lemma is essential and proved in \cite{js12} but we give
an independent proof in the Appendix (see also \cite{ms13}).

\begin{lemma}
Let $\bM=[m_{ab}]_{(a,b)\in\CA^2}$ be a matrix such that 
$m_{ab}\ge 0$. We assume that $1$ is the 
largest eigenvalue of $\bM$. Let $\bQ$ be a matrix with coefficients 
$q_{ab}=e^{i\theta_{ab}}m_{ab}$ where $\theta_{ab}$ is real. The
matrix $\bQ$ has eigenvalue 1 if and only if $\bQ$ is imaginary 
conjugate to matrix $\bM$.
\label{lem-conj}
\end{lemma}

\begin{corollary}
Let $c\in\CA$. The matrix $\bQ$ defined in Lemma~\ref{lem-conj}
has eigenvalue 1 if and only if for all 
$(a,b)\in\CA^2$:
\be
\frac{1}{2\pi}\left(\theta_{ab}+\theta_{ca}-\theta_{cb}\right)\in\mathbb{Z}~.
\ee
\label{coco}
\end{corollary}
\begin{proof}
If $Q$ is conjugate to $M$, we should have a real vector 
$\theta_{a\in\CA}$ such that $\forall (a,b)\in\CA^2$ $\theta_{ab}=\theta_a-\theta_b$. 
Then $e^{i(\theta_{a}-\theta_b)}=\frac{e^{i\theta_{cb}}}{e^{i\theta_{ca}}}$,
thus $e^{i(\theta_{cb}-\theta_{ca})}=e^{i\theta_{ab}}$.
\hfill
\end{proof}

\begin{lemma}\it
\label{theo-lam}
Let $c\in\CA$. 
A tuple $(s_1,s_2)$ belongs to $\partial\CK$ iff for all $(a,b)\in\CA^2$ we have 
\be
\frac{\Im(s_1)}{2\pi}\log\frac{P_1(a|b)P_1(c|a)}{P_1(c|b)}-\frac{\Im(s_2)}
{2\pi}\log\frac{P_2(a|b)P_2(c|a)}{P_2(c|b)}\in\mathbb{Z}~.
\ee
\label{lemDK}
\end{lemma}
\begin{proof}
Set $\bM=\bP(c_1,c_2)$ and $\bQ=\bP(s_1,s_2)$ for $(s_1,s_2)\in \partial\CK$. 
Then, it follows directly 
from Corollary~\ref{coco} with 
$e^{i\theta_{ab}}=(P_1(a|b))^{i\Im(s_1)}(P_2(a|b))^{-i\Im(s_2)}$. 
\end{proof}

Furthermore, in the Appendix we prove the following important 
characterization of the set $\CK$. We say that
a curve is strictly concave (or strictly convex) if the is never 
linear, even locally.

\begin{lemma}\it
If $\bP_1$ and $\bP_2$ are not conjugate, then the set $\CK$ is strictly concave.
\label{lemconv}\end{lemma}

We summarize our knowledge about $\partial \CK$.

\begin{theorem}
There are three possible structures of $\partial\CK$:
\begin{itemize}
\item the punctual case: $\partial\CK=\{(c_1,c_2)\}$, this is the most typical case;
\item the linear case: there exist a vector $(x,y)\in\mathbb{R}^2$ such 
that $\partial\CK=\{(c_1,c_2)+ik(x,y),k\in\mathbb{Z}\}$;
\item the lattice case: there exists two vectors $(x_1,y_1)$ and 
$(x_2,y_2)\in\mathbb{R}^2$ which are not colinear such that 
$\partial\CK=\{(c_1,c_2)+ik_1(x_1,y_1)+ik_2(x_2,y_2), (k_1,k_2)\in\mathbb{Z}^2.\}$.
\end{itemize}
\label{theoDK}
\end{theorem}

\begin{proof}
This follows from the fact that according to Lemma~\ref{lemDK} if 
$(c_1,c_2)+(s_1,s_2)\in\partial\CK$ then $\forall k\in\mathbb{Z}$ $(c_1,c_2)
+k(s_1,s_2)\in\partial\CK$. Furthermore if $(c_1,c_2)+(s_1',s_2')
\in\partial\CK$ then $(s_1,s_2)+a(s_1',s_2')\in\partial\CK$. Thus 
$\CK$ forms a lattice. In Lemma~\ref{lemDK} this occurs when 
$\bP_1$ and $\bP_2$ are rationally related. 

When both matrices $\bP_1$ and $\bP_2$ are logarithmically rationally 
related then we are in the lattice case, and the lattice is made of edges 
parallel to the axes. Anyhow the reverse is not necessarily true, 
although we don't know an explicit example of non logarithmically rationally 
related matrix which makes a pair of logarithmically commensurable 
matrices which would lead to edges non parallel to the axes. 

When only one matrix is  logarithmically rationally related, then we are in the linear case, 
and $\partial\CK$ is a set of periodic points laying on one axis. It is nevertheless 
possible to have a linear case when none of the matrices is logarithmically rationally 
related, for example when $\bP_1$ and $\bP_2$ are of the form
$\log^*\bP_1=2\pi\bQ_1+\bM$ and $\log^*\bP_2=-2\pi\bQ_2+\bM$ 
where $\bQ_1$ and $\bQ_2$ have integer coefficients but $\bM$ is not rationally 
related (in this case $x\log^*\bP_1+y\log^*\bP_2$ integers would implies $x=y$. 
\end{proof}

Now we establish some properties of the eigenvalue $\lambda(s_1,s_2)$ of
$\bP(s_1,s_2)$.

\begin{lemma}\it
For all $s_2$ such that $\Re(s_2)=c_2$, assume
$\not\exists s_1: (s_1,s_2)\in\partial\CK$ then 
$\lb(s_1,s_2)=1\Rightarrow \Re(s_1)<c_1$.
\label{lemL2}\end{lemma}
\begin{proof}
Notice that $\Re(s_1)=c_1$ is not possible by construction since it would imply 
that $(s_1,s_2)\in\partial\CK$. Let's consider the hypothesis $\Re(s_1)>c_1$. 
But we have $|\lb(s_1,s_2)|\le\lb(\Re(s_1),\Re(s_2)$. 
Since $\Re(s_1)>c_1$, each non zero coefficient of $\bP(\Re(s_1),\Re(s_2))$ 
are strictly smaller than the corresponding coefficients $\bP(c_1,\Re(s_2))$ and 
therefore $\lb(\Re(s_1),\Re(s_2))<\lb(c_1,\Re(s_2))=1$ which contradicts the hypothesis $\lb(s_1,s_2)=1$.
\end{proof}



\begin{lemma}\it
We have $\lb(c_1,c_2)>|\lb_2(c_1,c_2)|$.
\end{lemma}
\begin{proof} It follows from Perron-Frobenius that the main eigenvalue is unique.
\end{proof}

Let $\CU$ be a complex neighborhood  of 0 such that 
$\forall s\in\CU$: $|\lb(c_2+s)|>|\lb_2(c_2+s)|$. Therefore the function $\lb(c_2+s)$ is analytic.
In the Appendix we prove the following lemma.

\begin{lemma}\it
Let $(x_k,y_k)$ be a sequence of complex numbers such that $\lim_{k\to\infty}\Re(x_k,y_k)=(c_1,c_2)$ and $|\lb(x_k,y_k)|\to\lb(c_1,c_2)=1$. 
Then for all $(s_1,s_2)\in\CU$ we have
\be
\forall j:~~ \lim_{k\to\infty}\frac{\lb_j(x_k+s_1,y_k+s_2)}{\lb(x_k+s_1,y_k+s_2)}=
\frac{\lb_j(c_1+s_1,c_2+s_2)}{\lb(c_1+s_1,c_2+s_2)},
\ee
and the function $\lb(x_k+s_1,y_k+s_2)$ are all analytic and uniformly bounded functions 
on a complex neighborhood of $(0,0)$ such that
\begin{eqnarray}
\lim_{k\to\infty}\lb(x_k+s_1,y_k+s_2)&=&\lb(c_1+s_1,c_2+s_2)\\
\lim_{k\to\infty}\nabla \lb(x_k+s_1,y_k+s_2)&=&\nabla\lb(c_1+s_1,c_2+s_2)~.
\end{eqnarray}
where $\nabla f$ is the gradient of $f$.
\label{lem8}\end{lemma}

\section{Proof of Theorem~\ref{th-nilpotent}: Nilpotent Case}
\label{sec-nilpotent}

In this section we consider the case when
the matrix $\bP(s_1,s_2)$ is nilpotent, that is,
there exists $K$ such that $\bP^K(s_1,s_2)=0$ for all $(s_1,s_2)$. 
We first provide a simple derivation, and then "recover" it through the
Mellin approach. 

Notice that  for $z\in\mathbb{C}$ 
$$
1+z\langle\bone_C|(\bI-z\bP(0,0))^{-1}\bone\rangle=1+\sum_{k\le K}
z^{k+1}\langle \bone_C|\bP^k(0,0)\bone\rangle
$$
is the generating function that enumerates all the 
common words between the language of source 1 and the language of source
2, including the empty word. Let us call this set $\CW$. Observe that 
$\langle\bone_C|\bone\rangle$ enumerate the word of length 1, 
and $|\CW|=1+\langle\bone_C|(\bI-\bP(0,0))^{-1}\bone\rangle$ is the 
total number of such common words. Notice that such words are all of length 
smaller than $K$. Since the Markov source are stationary 
we also notice that $\vpi(0,0)=\bone_C$.

The quantity $J_{n,m}$ converges to 
$$1+\langle\bone_C|(\bI-\bP(0,0))^{-1}\bone\rangle$$ 
when $n,m\to\infty$ because all words in $\CW$ will appear in both 
string almost surely. Indeed each word in $w\in\CW$ may not appear in 
one string with exponentially small  probability.

For similar reasons $C_{n,m}$ will converge to  
$$1+\langle\bone_C|(\bI-\bP(0,0))^{-1}\bone\rangle$$ 
exponentially fast, because any word $w\in\CW$ may be prefix to none of $n$
independent strings with a probability decaying exponentially fast to 0. 

Interestingly enough we can find partially this result via the reverse 
Mellin transform (\ref{eq-dml}). Partially because the error term is 
$O(n^{-M})$ for all $M>0$. Let
\begin{eqnarray*}
D(s_1,s_2)&=&\langle\vpi(s_1,s_2)|
\tbC(s_1,s_2)+s_1\tbC(s_1,-1)\\
&&+s_1\hbC(-1,s_2)+s_1s_2\hbC(-1,-1)\rangle .
\end{eqnarray*}
We notice that $D(s_1,s_2)$ is never singular and furthermore 
for all $s$ $D(s,-1)=D(-1,s)=0$. Let
$$
D_n=\frac{1}{(2i\pi)^2}\int_{\rho_1}\int_{\rho_2} 
\Gamma(s_1)\Gamma(s_2)D(s_1,s_2)n^{-s_1-s_2}.
$$
Thus by (\ref{eq-dml}) we find 
$C(n,n)=(1-e^{-n})^2+D_n$. Let $M$ be an arbitrary non negative 
(large) number. By moving the integration path for 
$s_2$ from $\Re(s_2)=\rho_1$ to $\Re(s_2)=M$ we only met the poles  of 
$\Gamma(s_2)$ on $s_2=-1$ with residues 
$$\frac{1}{2i\pi}\int_{\rho_1}\Gamma(s_1)D(s_1,-1)n^{1-s_1}ds_1$$ 
and 
$$-\frac{1}{2i\pi}\int_{\rho_1}\Gamma(s_1)D(s_1,0)n^{-s_1}ds_1.
$$
The first residues is null since $D(s_1,-1)=0$, thus
\begin{eqnarray*}
D_n&=&-\frac{1}{2i\pi}\int_{\rho_1}\Gamma(s_1)D(s_1,0)n^{-s_1}ds_1\\
&&+\frac{1}{(2i\pi)^2}\int_{\rho_1}\int_{M} \Gamma(s_1)\Gamma(s_2)n^{-s_1-s_2}
\end{eqnarray*}
where the second term in the right-hand side is $O(n^{-M-\rho_1})$. The integration path
$-\frac{1}{2i\pi}\int_{\rho_1}\Gamma(s_1)D(s_1,0)n^{-s_1}ds_1$ can also be 
moved on $\Re(s_1)=M$, the residues on $s_1=-1$ is $D(-1,0)n$, 
which is null, and on $s_1=0$ is equal to $D(0,0)$. Thus
\be
D_n=D(0,0)-\frac{1}{2i\pi}\int_{M}\Gamma(s_1)D(s_1,0)n^{-s_1}+O(n^{-M-\rho_1}).
\ee
Since $\frac{1}{2i\pi}\int_{M}\Gamma(s_1)D(s_1,0)n^{-s_1}=O(n^{-M})$ 
and that $D(0,0)=\langle\bone_C|(\bI-\bP(0,0))^{-1}\bone\rangle$,
this  concludes the proof.

\section{Special Case: Proofs of Theorems~\ref{theogenmark0} -- \ref{theo01}}
\label{sec-special}

To simplify our presentation we will first assume that 
$$\bP_1=\frac{1}{|\CA|}\bone\otimes\bone,$$
{\it i.e.} the first source 
is uniform and memoryless. We will see in the next section
how to translate these results into the general case.

In this case, we have
\be
\bP(s_1,s_2)=|\CA|^{s_1}\bP(s_2)
\ee
with $\bP(s)=\bP(0,s)$.
We also write $\pi(a)=\pi_2(a)$ and  
$\bpi(s)=\bpi(0,s)$, thus
\be
\bpi(s_1,s_2)=|\CA|^{s_1}\bpi(s_2)~.
\ee

Let $\lb(s_1,s_2)$ be the main (largest) eigenvalue of $\bP(s_1,s_2)$. We have
\be
\lambda(s_1,s_2)=|\CA|^{s_1}\lambda(s_2)
\ee 
where $\lb(s)$ is the main eigenvalue of matrix $\bP(s)$. We also define
$\bun(s)$ as the right eigenvector of $\bP(s)$ and 
$\bzeta(s)$ as the left eigenvector
provided $\la\bzeta(s)|\bun(s)\ra=1$. 

We first present some simple results regarding $\lambda(s)$ and
$L(s)= \log_{|\CA|} \lambda(s)$. 

\begin{lemma}
The function $L(s)$ is convex when $s$ is real.
\label{lemLs}
\end{lemma}
\begin{proof}
The function $(-L(s),s)$ describes the set $\CK$ 
which is known to be a concave curve by Lemma~\ref{lemconc}. 
Notice that the proof will also be valid for the general case.
\end{proof}
The proof of the following lemma is left for the reader.

\begin{lemma}
We have the following identities:
\be
\begin{array}{rcl}
\lb(s)&=&\la\bzeta(s)|\bP(s)\bun(s)\ra=\sum_{a,b}\zeta_a(s)u_b(s)P(a|b)^{-s},\\
\lb'(s)&=&\la\bzeta(s)|\bP'(s)\bun(s)\ra=\sum_{a,b}\zeta_a(s)u_b(s)P(a|b)^{-s}
(-\log P(a|b)),\\\
\lb''(s)&=&\la\bzeta(s)|\bP''(s)\bun(s)\ra=\sum_{a,b}\zeta_a(s)u_b(s)
P(a|b)^{-s}(\log P(a|b))^2.
\end{array}
\ee
\label{lem-lb}
\end{lemma}

Finally, to compute some of the constants in 
Theorems~\ref{theogenmark0} -- \ref{theo01},
we need to computer $L''(s)$. To do so, 
let $x_{a,b}=\frac{1}{\lb(s)}\zeta_a(s)u_b(s)P(a|b)^{-s}$. Clearly, by Lemma~\ref{lem-lb} we have
$\sum_{a,b}x_{a,b}=1$ and
\be
L''(s)=\sum_{a,b}x_{a,b}(\log P(a|b))^2-\left(\sum_{a,b}x_{a,b}\log P(a|b)
\right)^2~.
\ee

Now we are ready to derive our results presented in  
Theorems~\ref{theogenmark0} -- \ref{theo01}. The starting point is the
Mellin transform $C^*(s_1,s_2)$ shown in (\ref{eth1}) with
$\tbC^*(s_1,s_2)$ presented in (\ref{purdue1}). 
To recover $C_{n,n}$ we first need to find the inverse Mellin transform of
(\ref{eth1}).  For $-2<\rho<-1$ we have
\begin{eqnarray*}
\label{eq-inv}
&&\bC(z,z)-\bD(z,z)=\\
&&\frac{1}{(2i\pi)^2}\int_{\Re(s_1)=\Re(s_2)=\rho}\tbC^*(s_1,s_2)
z^{-s_1-s_2}ds_1ds_2 \\
&=&\frac{1}{(2i\pi)^2}\int_{\Re(s_1)=\Re(s_2)=\rho}\Gamma(s_1)\Gamma(s_2)
\left(\hbC(s_1,s_2)+s_2\hbC(s_1,-1)\right.\\
 &&\left.+s_1\hbC(-1,s_2)+s_1s_2\hbC(-1,-1)\right)
z^{-s_1-s_2}ds_1ds_2,
\end{eqnarray*}
where $\hbC(s_1,s_2)=\left(\bI-\bP(s_1,s_2)\right)^{-1}\bone$ and  
$$
\tbC(s_1,s_2)=\langle\bpi(s_1,s_2)\left(\bI-\bP(s_1,s_2)\right)^{-1}\bone\rangle. 
$$
Since $\bC(z,z)=\bD(z,z)+O(z^{-M})$ for any $M>0$ when $\Re(z)\to\infty$ we  find
\begin{eqnarray}
C(z,z)&-&1+O(z^{-M})=\frac{1}{(2i\pi)^2}\int_{\Re(s_1)=\Re(s_2)=\rho}
\Gamma(s_1)\Gamma(s_2) \label{eq-m1} \\
&&\times\left(\hbC(s_1,s_2)+s_2\hbC(s_1,-1)\right. \nonumber\\
&&\left. s_1\hbC(-1,s_2)+s_1s_2\hbC(-1,-1)\right) z^{-s_1-s_2}ds_1ds_2.
\nonumber
\end{eqnarray}
To analyze it asymptotically, we investigate the set of singularities of 
$(\bI-\bP(s_1,s_2))^{-1}$ in $\hbC(s_1,s_2)$. 
Recall that $\CK$ is the set of complex numbers 
$(s_1,s_2)$ such that $\bI-\bP(s_1,s_2)$ is degenerate, 
{\it i.e.} $(\bI-\bP(s_1,s_2))^{-1}$ is singular. 

Let $\lb_1(s),\lb_2(s)\ldots,\lb_{|\CA|}(s))$ be the eigenvalues 
of $\bP(s)$ in the non-increasing order 
(e.g., $\lb(s):=\lb_1(s)$) while 
$\bun_i(s)$ and $\bzeta_i(s)$) are respectively the right and the left 
eigenvectors of $\bP(s)$ associated with $\lb_i(s)$ subject to
$\la\bzeta_i(s)|\bun_i(s)\ra=1$.
By the spectral representation of matrices \cite{spa-book}, we  have
\be
(\bI-\bP(s_1,s_2))^{-1}=\sum_{i=1}^{|\CA|}\frac{1}{1-|\CA|^{s_1}\lb_i(s_2)}
\bun_i(s_2)\otimes\bzeta_i(s_2)
\label{eq-spectral}
\ee
where $\otimes$ denotes the tensor product.
Observe that $(\bI-\bP(s_1,s_2))^{-1}$ cease to exist at $(s_1,s_2)$
satisfying $|\CA|^{s_1} \lambda_i(s_2)=1$, that is, for $s_1:=L_{i,k}(s_2)$
where
$$L_{i,k}(s_2)=\frac{1}{- \log|\CA|}(\log\lb_i(s_2)+2ik\pi).$$

The eigenvalues $\lb_i(s)$ are individually analytic functions of $s$ in any complex neighborhood where the order of the eigenvalues modulus does not change 
({\it i.e.} $|\lb_{i-1}(s)|>|\lb_{i}(s)|>|\lb_{i+1}(s)|$ for all $i$). 
But any function of the form $\sum_i f(\lb_i(s))$ is analytic even when the eigenvalue sequence is not strictly
decreasing, as long as $f()$ is analytic.
To simplify our analysis, we also postulate that none of the eigenvalue is identically 
equal to zero, that is, we assume  $\log\lb_i(s)$ exists except on a
countable set $\CR=\{s: ~ \exists i: \lb_i(s)=0\}$. 
It should be pointed out that there are cases when some eigenvalues are
identically equal to zero. For example, for memoryless sources we
have for all $i\ge 2$: $\lb_i(s)\equiv 0$ which we already discussed in 
\cite{jacquet,js-book} so we will omit them here.

In order to evaluate the integral in (\ref{eq-m1}) we first use (\ref{eq-spectral}) and
then apply the residue theorem. To simplify,
for $1\leq j \leq |\CA|$, let $f_j(s)=\la\bpi(s)|\bun_i(s)\ra$ and 
$g_j(s)=\la\bzeta_i(s)|\bone\ra$. Define  (here we set $s:=s_2$)
\begin{eqnarray}
\label{eq-Lk}
&&I(z,\rho)=\\
&&\frac{1}{2i\pi}\int_{\Re(s)=\rho}\sum_{k\in \mathbb{Z}}
\sum_{j=1}^{|\A|}
\frac{f_j(s)g_j(s)\Gamma(-L_{j,k}(s))\Gamma(s)}{\lb_j(s)
\log|\CA|}z^{L_{j,k}(s)-s}ds~.\nonumber\\
&&J(z,\rho)=\frac{1}{2i\pi}\int_{\Re(s)=\rho}\hC(0,s)\Gamma(s)z^{-s}ds.
\end{eqnarray}
Furthermore, let
\be
\label{purdue2}
H_j(s,z)=\sum_{k\in \mathbb{Z}}\frac{f_j(s)g_j(s_2)
\Gamma(-L_{j,k}(s))}{\lb_i(s)\log|\CA|}z^{L_{j,k}(s)}
\ee
thus 
$$
I(z,\rho)=\frac{1}{2i\pi}\sum_{j}\int_{\Re(s)=\rho}H_j(s,z)\Gamma(s)z^{-s}ds. 
$$

The next lemma is crucial for the asymptotic evaluation of $C(z,z)$ which by
depoissonization lead to asymptotics of $C_{n,n}$ and ultimately $J_{n,n}$.

\begin{lemma}
For any $M>0$ and for some $\rho>-1$, we have 
\be
C(z,z)= 1+I(z,\rho)-J(z,\rho)+O(z^{-M})
\label{eq-rho}
\ee
for $z \to \infty$.
\label{lem-rho}\end{lemma}

\begin{proof}
In the inverse Mellin expression we see that for 
$\Re(s_1)=\Re(s_2)=-1$ we have $|\bP(s_1,s_2)|\le\bP(-1,-1)$ 
and $\bP(-1,-1)\le \bP_1(-1)$ and $\bP_2(-1)$. 
Since the matrix $\bP(s_1,s_2)$ is not nilpotent there exists 
$(a,b)\in\CA^2$ such that $|\bP(s_1,s_2)_{a,b}|<\bP(a|b)$. Consequently, 
there exists $k$ such that $|\bP^k(s_1,s_2)|\bone<\bone$ or more 
precisely $|\lb_1(s_1,s_2)|\le 1-\eps'$ for some $\eps'>0$. 
Thus there exists $\eps>0$ such that for all $s_1,s_2$ $\Re(s_1)>-1+\eps$ 
and $\Re(s_2)>-1+\eps$ implies that $\bI-\bP(s_1,s_2)$ is not degenerate. 

To evaluate the inverse Mellin transform we apply standard approach by moving the
line of integration to "catch up" relevant singularities, however, in our case
there some complications.
We move the integration path by increasing $\rho$. This does not change 
the value of $I(z,\rho)$ and $J(z,\rho)$ as long as the functions in the 
integral paths are analytic and not singular. When the path encounter a 
singularity we will use the residue theorem. But we may have a problem 
when any of the functions $\lb_j$ ceases to be analytic. 
However, we shall see that when we sum all the terms of the integrand of 
$I(z,\rho)$ we obtain an analytic function derived from 
$\langle\bpi(s_1,s_2)(\bI-\bP(s_1,s_2))^{-1}\bone\rangle$. 
Indeed we have the (somewhat complicated) identity
\be
I(z,\rho)=\sum_{k\in\mathbb{Z}}\int_{\Re(s)=\rho}\la\bpi(s)|(\bP(s))^{-1}\exp\left(-\frac{\log z}{\log|\CA|}
(\log\bP(s)+2ik\pi\bI)\right)\Gamma\left(-\frac{1}{\log|\CA|}
(\log\bP(s)+2ik\pi\bI)\right)\bone\ra ds~,
\ee
knowing that any analytical function $f(.)$ can be applied to matrix 
$\bP(s)$ as long its eigenvalues do not correspond to a singularity of 
the function $f(.)$. 
Therefore the only singularities that we meet when we move the integration 
line of $I(z,\rho)$ are the elements of $\CR=\{s: ~ \lb_i(s)=0, \mbox{for~some~}i \}$.

If $\theta\in\CR$, is one of these singularity, thus we have $\lb_i(\theta)=0$, 
then the function 
$L_{i,k}(s)=\frac{1}{- \log|\CA|}(\log\lb_i(s)+2ik\pi)$ is meromorphic around $\theta$. 
However if $\theta$ is a simple pole of $\lb_i(s)$, then moving around $\theta$ 
would be equivalent to add 1 to the integer $k$: $\log\lb_i(s)\to\log\lb_i(s)+2i\pi$. 
If the root is of multiplicity $\ell$ it is equivalent to add $\ell$ to the integer $k$.
In any case the function $H_i(s,z)$ being invariant when $\ell$ is added to $k$, 
turns out to be fully analytic around $\theta$, and the integration path in 
$I(z,\rho)$ can be moved over $\theta$.

However, the function $\lb_i(s)$ is a non polar singular on $s=\theta$, hence there 
will be a contribution coming from the integration of $H_i(s,z)\Gamma(s)z^{-s}$ on 
an arbitrary small loop around $\theta$. 
Since $\Re(L_{i,k}(s))\to-\infty$ when $s\to\theta$, having $\Re(L_{i,k}(s))<-M$ 
will guarantee that the contribution is in $O(z^{1-M})$ and can be 
included in the error term.

Moving the integration path from $\Re(s_1)=\Re(s_2)=\rho$ to $\Re(s_1)=
\Re(s_2)=-1+\eps$ will only hit the poles of $\Gamma(s_1)\Gamma(s_2)$ 
at $s_1=-1$ and $s_2=-1$. By construction of the function 
$\hbC(s_1,s_2)+s_1\hbC(-1,s_2)+s_2\hbC(s_1,-1)+s_1s_2\hbC(-1,-1)$, 
the residues at these points are zero. Therefore the expression 
\begin{eqnarray*}
C(z,z)&-&1+O(z^{-M})=\frac{1}{(2i\pi)^2}\int_{\Re(s_1)=
\Re(s_2)=\rho}\Gamma(s_1)\Gamma(s_2)\\
&&\times\left(\hbC(s_1,s_2)+s_2\hbC(s_1,-1)\right.\\
&&\left. s_1\hbC(-1,s_2)+s_1s_2\hbC(-1,-1)\right) z^{-s_1-s_2}ds_1ds_2.
\end{eqnarray*}
still holds for $\rho=-1+\eps$. 

Now we take the integration contour for $s_1$ and we move it 
from $\Re(s_1)=\rho$ to $\Re(s_1)=M-\rho$. 
By doing so we encounter many poles:
\parsec
(i) The poles of $(\bI-\bP(s_1,s_2))^{-1}$ at $s_1=L_{j,k}(s_2)$. 
The residues is exactly the expression $I(z,\rho)$.
\parsec
(ii) The poles of $\hC(s_1,s_2)\Gamma(s_1)$  at $s_1=0$ which has 
residues $-J(z,\rho)$. 
\parsec
(iii) The double pole of $\hC(s_1,-1)\Gamma(s_1)z^{-s_1}$ at $s_1=0$ 
since $(\bI-\bP(s_1,s_2))^{-1}$ is singular at $(s_1,s_2)=(0,-1)$ 
because $\lb_1(0,-1)=1$. It leads to the residue
\be
-\frac{1}{2i\pi}\int_{\Re(s)=\rho} s\Gamma(s)(a\log z+b)z^{-s}ds
\label{eq-Re}
\ee
for some real number $a$ and $b$ coming from the derivative of 
$f_1(s)$ and $g_1(s)$ at $s=0$. But when one moves the integration path 
of~(\ref{eq-Re}) to $\Re(s)=M$ the function $s\Gamma(s)(a\log z+b)$ 
has no singularity since $s\Gamma(s)$ is not singular on the 
interval $]-1,+\infty[$, and thus the integration on 
$\Re(s)=M$ is $O(z^{-M})$, which can be included in the error term. 
\end{proof}


In the following we denote $L(s)=L_{1,0}(s)$ .
The rule of the game is that we move the integration abscissa of 
$I(z,\rho)$ and $J(z,\rho)$ to the left ({\it i.e.} to larger values) 
on the value $c_2$ which minimizes the argument $L(s)-s$. Moving the 
integration path one meets some poles of $\Gamma(-L_j(s))$ when $L_j(s)=0$. 
In fact when the matrices are strictly non negative, this case only 
applies to $j\ge 2$. It turns out that when $s$ is a pole for 
$\Gamma(-L_j(s))$ then it is at the same time a pole of $\hC(0,s)$. 
The residues of $I(z,\rho)$ and $J(z,\rho)$ when passes over such 
value are the same and cancel. Therefore 
$C(z,z)=1+I(z,\rho)-J(z,\rho)+O(z^{-M})$ for all values of $\rho<0$.

\subsection{Proof of Theorem~\ref{theogenmark0} and  \ref{theo01} }

Now we are going to prove Theorem~\ref{theogenmark0} and  \ref{theo01} 
corresponding to the case where quantities $c_1$ and $c_2$ are both in the 
interval $[-1,0]$ in the case where one source is uniform memoryless. 
In this case, the main contribution to $C(z,z)$ doesn't come from the
poles, as in the previous section, but rather from the saddle point of
$z^{L_{1,0}-s}$ (in fact, infinitely many saddle points). 

We start with reviewing some properties of the kernel $\CK$ and the main
eigenvalue.  Recall that $\partial\CK$ is the set of complex tuples 
$(s_1,s_2)$ satisfying $|\CA|^{s_1}\lb(s_2)=1$ such that 
$\Re (s_1)=c_1$ and $\Re(s_2)=c_2$.
Its structure is crucial for our asymptotic analysis. 

>From the general Theorem~\ref{theoDK} we deduce that only two cases 
are possible when one source, say source 1, is uniform 
(since $\bP_1$ is logarithmically rationally related): 
\begin{itemize}
\item the lattice case when $\bP_2$ is also logarithmically rationally 
related, we call this case the {\it rational case};
\item the linear case when $\bP_2$ is not logarithmically rationally related, 
we call this case the {\it irrational case};
\end{itemize}

Now we focus on proving in the next four lemmas that the main eigenvalue is well
separated.

\begin{lemma}
Let $t_2$ be a real number.  We have the equivalence
$$
\nexists (s_1,s_2)\in\partial\CK ~ ~ \Im(s_2)=t_2 ~
\iff ~ |\lb(c_2+it_2)|<\lb(c_2).
$$
\end{lemma}
\begin{proof}
Let $s_2=c_2+it_2$. By the Perron-Frobenius, we have $|\lb(s_2)|\le\lb(c_2)$ 
since $\Re(s_2)=c_2$ and $|\bP(s_2)|=\bP(c_2)$ (by taking the modulus 
element-wise). If $|\lb(s_2)|=\lb(c_2)$, then there will be $t_1$ such 
that $|\CA|^{it_1}\lb(s_2)=\lb(c_2)$, and therefore 
$(c_1+it_1,s_2)\in\partial\CK$.
\end{proof}

\begin{lemma}
We have a non zero spectral gap, that is,  $\lb(c_2)>\lb_2(c_2)$.
\end{lemma}
\begin{proof} It follows from Perron-Frobenius that the main eigenvalue is unique.
\end{proof}

Let $\CU$ be a complex neighborhood  of 0 such that 
$\forall s\in\CU$: $|\lb(c_2+s)|>|\lb_2(c_2+s)|$. 
Therefore the function $\lb(c_2+s)$ is analytic.

\begin{lemma}
Let $s_k$ be a sequence such that $\Re(s_k)=c_2$ and $|\lb(s_k)|\to\lb(c_2)$. 
Then for all $s\in\CU$ we have
\begin{eqnarray}
\lim_{k\to\infty}L(s_k+s)-L(s_k)&=&L(c_2+s)-L(c_2)\\
\lim_{k\to\infty}L'(s_k+s)&=&L'(c_2+s)~.
\end{eqnarray}
The convergence also holds for any derivative of function $L'(s)$, 
and the function $\lb(s_k+s)$ is analytic and uniformly bounded 
on a complex neighborhood of $0$.
\label{lem8bis}\end{lemma}
\begin{proof}
It turns out that $\lim_k|\lb(c_1,s_k)|=|\CA|^{c_1}\lb(c_2)=1$. There 
exists $x_k$ such that $\Re(x_k)=c_1$ and $\lb(x_k,s_k)=|\lb(c_1,s_k)|$.
Hence Lemma~\ref{lem8} applies. Thus for any complex number $s$
\begin{eqnarray}
\forall j:~~ \lim_{k\to\infty}\frac{\lb_j(s_k+s)}{\lb(s_k+s)}&=&\lim_{k\to\infty}\frac{\lb_j(x_k,s_k+s)}{\lb(x_k,s_k+s)}\nonumber\\ 
&=&\frac{\lb_j(c_1,c_2+s)}{\lb(c_1,c_2+s)}=\frac{\lb_j(c_2+s)}{\lb(c_2+s)}
\end{eqnarray}
Since $\frac{|\lb_2(c_2)|}{\lb(c_2)}<1$, there exists $\CU$ such that 
$\forall s\in\CU$: $\left|\frac{\lb_j(c_2+s)}{\lb(c_2+s)}\right| <1$ 
thus $\lb(c_2+s)$ is analytic because it never cross the value of 
another eigenvalue and so is $\lb(s_k+s)$.

Hence, the logarithm of the eigenvalue, $L(s_k+s)-L(s_k)$ converges to 
$L(c_2+s)-L(c_2)$. The property $|\lb(c_1+s,c_2+s_2)|>\lb_2(c_2+s)$ for all 
$s\in\CU$ implies the analyticity of $L(c_2+s)$, and therefore $L'(c_2+s)$.
\end{proof}

In passing, we have $L'(s_k)\to1$ and $L''(s_k)\to\alpha_2$.

Finally, we prove that the main eigenvalue dominates  all other eigenvalues in
a complex  neighborhood of $c_2$.

\begin{lemma}
\label{lem9}
There exists $\eps>0$ such that for all $i\neq 1$ and for all $s$ 
such that $\Re(s)=c_2$ :
\be
|\lb_i(s)|<\lb(c_2)-\eps~.
\ee
\end{lemma}
\begin{proof}
This is a consequence of previous lemmas. 
Suppose that there exists $s_k$ such that $|\lb_2(s_k)|\to\lb(c_2)$. 
This implies that $|\lb(s_k)|\to\lb(c_2)$, but by previous 
lemma $|\lb_2(s_k)|\to\lb_2(c_2)=\lb(c_2)-\eps$. 
\end{proof}

Now we are in the position to evaluate the integral of $C(z,z)$ by the
saddle point methods.
Recall that for all $M>0$ we have  $C(z,z)=I(z,c_2)-J(z,c_2)+1+O(z^{1-M})$ where
$I(z,\rho)$ and $J(z,\rho)$ are given by (\ref{eq-Lk}). We already 
prove that $J(z,c_2)=O(z^{c_2})=O(z^{\kappa-\eps})$ for some $\eps>0$ 
(in fact $\eps=L(c_2)>0$). We reinforce it in the next lemma.

\begin{lemma}
There exists $\eps>0$ such that
\be
\label{eq-jac}
C(z,z)=1+\frac{1}{2i\pi}\int_{\Re(s)=c_2}H_1(s,z)\Gamma(s)z^{-s}ds+O(z^{\kappa-\eps})
\ee
where $\kappa=-c_1-c_2$ and we recall that
$$
H_1(z,s)=\sum_{k\in\mathbb{Z}}\frac{f_1(s)g_1(s)
\Gamma(-L_{1,k}(s))}{\lambda(s)\log|\CA|}. 
$$
\label{lem18}\end{lemma}
\begin{proof}
By Lemma~\ref{lem9} for all $j>1$ we have
$\log_{|\CA|}|\lb_j(s)|<\log_{|\CA|}\lb(c_2)-\eps$
for some $\eps>0$,  thus the contribution of $\int_{\Re(s)=c_2}H_j(s,z)
\Gamma(s) z^{-s}ds$ is of order 
$$
\int_{\Re(s)=c_2}|\Gamma(s)| z^{\Re(L_j(s)-s)}ds =
O\left(z^{L(c_2)-c_2-\eps}\right)=O\left(z^{\kappa-\eps}\right),
$$
as desired.
\end{proof}

\med
{\bf Rational Case}.
We assume now that the matrix $\log^*(\frac{1}{P(c|c)}\bP)$ is 
{\it rationally balanced}. The matrix
$\bP(s+2i\pi\nu)$ is then imaginary conjugate with the matrix
$P(c|c)^{2i\pi\nu}\bP(s)$ and
$L(s+2i\pi\nu)=L(s)+2i\pi\nu\log P(c|c)$. Thus $\Re(L(c_2+it))$
is periodic in $t$ with period $2\pi\nu$. 
Furthermore, $L'(s)$ is also periodic with period $2\pi\nu$. 
Thus, $s_\ell=c_2+2i\pi\ell\nu$ for $\ell\in\mathbb{Z}$ are
saddle points of $z^{L(s)-s}$. 

We concentrate now on the term $k=0$ in $H_1(s,z)\Gamma(s)z^{-s}$ in (\ref{eq-jac}).
Define
\be
b_2(s)=\frac{d^2}{ds^2}\log\left(\frac{f_1(s)g_1(s)}{\lb_1(s)}\Gamma(-L(s))\Gamma(s)
\right).
\ee
Notice that $b_2(s)=\beta_2(-L(s),s)$ mentioned in Theorem~\ref{theogenmark}.
Since the function
$$
\log\left(\frac{f(s)g(s)}{\lb(s)}\Gamma(-L(s))\Gamma(s)\right)
$$
has bounded variations, we have the classic saddle point result
\cite{fs-book,spa-book}
$$
\frac{1}{2i\pi}\int_{\Re(s)=c_2}\frac{f(s)g(s)}{\lb(s)}\Gamma(-L(s))
\Gamma(s)z^{L(s)-s}ds=
$$
\begin{eqnarray}
&=&\sum_\ell\frac{f(s_\ell)g(s_\ell)}{\lb(s_\ell)}
\Gamma(-L(s_\ell))\Gamma(s_\ell)\nonumber\\ 
&&\times\frac{z^{L(s_\ell)-s_\ell}}
{\sqrt{2\pi(\alpha_2\log z+b_2(s_\ell))}}(1+o(1)).
\label{eq-saddle}\end{eqnarray}
Notice that $\Re(L(s_\ell)-s_\ell)=\kappa$. When adding the contribution
from the $L(s)+\frac{2ik\pi}{\log|\CA|}$ we obtain the expression for
$Q(\log z)$ with $\partial\CK=\{(-L(s_\ell)-\frac{2ik\pi}
{\log|\CA|},s_\ell),(k,\ell)\in\mathbb{Z}^2\}$. The double periodicity 
comes from the fact that $\sqrt{x}Q(x)=\sum_{k,\ell}q_{k,\ell}
e^{i(k\alpha+\ell \beta)x}+o(1)$ when $x\to\infty$ for some 
{\it incommensurable}\footnote{recall that a pair of numbers $(\alpha,\beta)$ is
{\it commensurable} if there exists a real number $\nu$ such that 
the vector $(\nu\alpha,\nu \beta) \in\mathbb{Z}^2$; 
otherwise the pair is {\it incommensurable}.}
pair of real numbers $(\alpha,\beta)$ and  
complex numbers  $\{q_{k,\ell}\}_{(k,\ell)\in\mathbb{Z}^2}$. 

\med
{\bf Irrational Case}.
We now turn to the irrational case.
Let $A>0$ be a number such that for all $|s|\le A$ we have 
$|\lb(c_2+s)|>|\lb_2(c_2+s)|$; thus $L(c_2+s)$ is analytic. 
We assume that $c_2<0$ is the only saddle point 
on $\Re(s)=c_2$ for $|\Im(s)|\le A$.
There also exists $\alpha_3>0$ such that 
\be
\label{eq-jac2}
|t|\le A\Rightarrow ~ ~ \Re(L(c_2+it)-L(c_2))\le-\alpha_3 t^2~.
\ee
>From the previous analysis we know that
\begin{eqnarray}
\frac{1}{2i\pi}\sum_{k\in\mathbb{Z}}&&\int_{\Re(s)=c_2,|\Im(s)|\le A}
\frac{f(s)g(s)}{\lb(s) \log |\CA|}\nonumber\\
&&\times\Gamma\left(-L(s)-\frac{2ik\pi}{\log|\CA|}\right)\nonumber\\
&&\times\Gamma(s)z^{L(s)-s+2ik\pi/\log|\CA|}ds\nonumber\\
&=&Q(\log z)(1+o(1)).
\label{eq-sadA}
\end{eqnarray}
Assume now (\ref{eq-jac2}) and define
\be
\xi(s)=\sum_{k\in\mathbb{Z}}\left|\Gamma\left(s-\frac{2ik\pi}{\log|\CA|}\right)
\right|~.
\ee
The function $\xi(s)$ is continuous and bounded as long as $\Re(s)$ is bounded.
Our aim is to prove that
\begin{eqnarray}
\frac{1}{2i\pi}&&\int_{\Re(s)=c_2,|\Im(s)|>A}\left|\frac{f(s)g(s)}{\lb(s)}\right| \xi(-L(s))\nonumber\\
&&\times|\Gamma(s)|z^{\Re(L(s))-c_2}ds=o(\frac{z^\kappa}{\sqrt{\log z}})~,
\end{eqnarray}
which will complete the proof of Theorem~\ref{theogenmark0}.

We know that $|f(s)g(s)|\le f(c_2)g(c_2)$. In addition,  
we know that for $\Re(s)=c_2$ we have $\Re(L(s))<L(c_2)$ 
as long as $\Im(s)\neq 0$. 
We also have $|\lb(s)|>\eps'$ for some $\eps'>0$ since the 
matrix $\bP(s)$ stays away from the null matrix.
Therefore, we need to estimate 
\be
\int_{\Re(s)=c_2,|\Im(s)|>A}|\Gamma(s)|z^{\Re(L(s))-c_2}ds~.
\ee
For any $\eps>0$, the portion of the line $\Re(s)=c_2$, where 
$\Re(L(s))<L(c_2)-\eps$, contributes $z^{\kappa-\eps}$ to $C(z,z)$. 
Our attention must turn to the values of $s$ on this line such 
that $\Re(L(s))$ is arbitrary close to $L(c_2)$. 
In particular, we are interested in the local maxima of $\Re(L(s))$ 
that are arbitrary close to $L(c_2)$.
Indeed, these local maxima play a role in the saddle point method.
 
Let us consider the sequence of those maxima denoted by $s_\ell$ for 
$\ell\in\mathbb{N}$ such that  $\Re(L(s_\ell))\to L(c_2)$. By 
Lemma~\ref{lem8} we know that for all real $t$ 
$L(s_\ell+it)-L(s_\ell)\to L(c_2+it)-L(c_2)$ and that 
$L'(s_\ell+it)\to L'(c_2+it)$. Therefore for 
all real $t$ such $|t|\le A$
\be
\lim\sup_{\ell\to\infty}(\Re(L(s_\ell+it))-\Re(L(s_\ell)))\le-\alpha_3 t^2.
\ee 
We define $I(A)$ to be the set of complex numbers $s$ such that 
$\Re(s)=c_2$ and $\min_\ell\{|s-s_\ell|\}>A$.
\begin{lemma}\it
There exists $\eps$ such that for all $s\in I(A)$: 
$\Re(L(s))<L(c_2)-\eps$.
\end{lemma}
\begin{proof}
Assume $s\in I(A)$. 
Since $s$ is not a local maxima, we study the variation of 
$\Re(L(s))$ around the local maxima $s_\ell$. 
Without loss of generality we assume that $s_\ell-A$ is between $s$ and 
$s_\ell$, thus $\Re(L(s_\ell-A))>\Re(s)$. 
Since $\lim\sup\Re(L(s_\ell-A))<L(c_2)-\alpha_3 A^2<L(c_2)-\eps$ the lemma 
is proven. 
\end{proof}

In view of the above, we conclude that
\begin{eqnarray}
\int_{\Re(s)=c_2,|\Im(s)|>A}|\Gamma(s)|z^{\Re(L(s))-c_2}ds\le\nonumber\\
\sum_\ell
\int_{|t|\le A}|\Gamma(s_\ell+it)|z^{\Re(L(s_\ell+it))-c_2}dt+O(z^{\kappa-\eps})~.
\end{eqnarray}
By virtue of the properties of function $\Gamma(s)$ on the imaginary lines, there exists a real $B>0$ such that $\forall s$:
\be
\Re(s)=c_2\Rightarrow \max_{|t|\le A}\{|\Gamma(s+it)|\}\le B|\Gamma(s)|~.
\ee
Therefore, our analysis can be limited to 
\be
\sum_\ell\int_{|t|\le A}|\Gamma(s_\ell)|z^{\Re(L(s_\ell+it))-c_2}dt~.
\ee

Finally, we establish a separation result.

\begin{lemma}\it
For $\ell$ tending to infinity, the $s_\ell$ are separated by a distance at 
least equal to $A$.
\end{lemma}
\begin{proof}
First, let us assume that $\ell,\ell'\to\infty$ and 
$|s_{\ell}-s_{\ell'}|\to 0$, then we have 
\be
L'(s_{\ell'})=L'(s_\ell)+(s_{\ell'}-s_\ell)L''(s_\ell)+O(|s_\ell-s_{\ell'}|^2). 
\ee
Since $L''(s_\ell)\to\alpha_2\neq 0$, then we cannot have $L'(s_{\ell'})=1$, 
thus $s_{\ell'}$ cannot be a local maximum of $\Re(L(s))$. 
Second, if $\lim\inf|s_{\ell}-s_{\ell'}|>\eps$ for some $\eps>0$ with 
$|s_\ell-s_{\ell'}|<A$, then using the inequality
\be
\lim\sup\Re(L(s_{\ell'}))-\Re(s_\ell)\le-\alpha_3|s_\ell-s_{\ell'}|^2
<-\alpha_3\eps^2
\ee
we cannot have $\Re(L(s_{\ell'}))\to L(c_2)$.
\end{proof}

 From the above we conclude
\begin{eqnarray}
&&\int_{|t|\le A}z^{\Re(L(s_\ell+it))-c_2}dt=\nonumber\\
&&z^\kappa z^{\Re(s_\ell)-L(c_2)}
\int_{|t|\le A}
z^{\Re(L(s_\ell+it))-\Re(L(s_\ell)}dt~.
\end{eqnarray}

In summary, the consequence of the previous lemma is that since 
$\lim\sup_{\ell\to\infty}\Re(L(s_\ell+it))-
\Re(L(s_\ell)\le-\alpha_3 t^2$, we have \cite{spa-book}
\be
\lim\sup_{\ell\to\infty}\int_{|t|\le A}z^{\Re(L(s_\ell+it))-\Re(L(s_\ell)}dt
\le\frac{1}{\sqrt{\pi\alpha_3\log z}},
\ee
and the properties of function 
$\Gamma(s)$ is that $\sum_\ell |\Gamma(s_\ell)|<\infty$. Therefore, 
\begin{eqnarray}
\sum_\ell&&\int_{|t|\le A}|\Gamma(s_\ell)|z^{\Re(L(s_\ell+it))-c_2}dt\nonumber\\
&=&z^\kappa\sum_\ell|\Gamma(s_\ell)|z^{\Re(s_\ell)-L(c_2)}\nonumber\\
&&\times\int_{|t|\le A}
z^{\Re(L(s_\ell+it))-\Re(L(s_\ell)}dt
\end{eqnarray}
since $\lim_{z\to\infty}z^{\Re(s_\ell)-L(c_2)}=0$, by
the dominating convergence theorem. We finally arrive at
\begin{eqnarray}
&&\sum_\ell|\Gamma(s_\ell)|z^{\Re(s_\ell)-L(c_2)}\int_{|t|\le A}
z^{\Re(L(s_\ell+it))-\Re(L(s_\ell)}dt\nonumber\\
&&=o(\frac{1}{\sqrt{\log z}})~.
\end{eqnarray}

In fact, the saddle point expansion is extendible to any order of 
$\frac{1}{\sqrt{\log n}}$. This proves Theorem~\ref{theogenmark0} and \ref{theo01} .
In passing we observe that when $\bP_1$ and $\bP_2$ 
are logarithmically commensurable, the line $\Re(s_1)=c_1$ contains an 
infinite number of saddle points that contribute to the double 
periodic function $Q_2(\log n)$ (cf. \cite{js12} for more details).


\section{General Case: Proofs of Theorem~\ref{theo11} -- \ref{theogege}}
\label{sec-last}

We now look at the general case when $\bP_1\neq\bP_2$. The main difficulty 
of the general case is when $\bP_1$ and $\bP_2$ have some zero coefficients, 
not at the same locations. For example $\bP(-1,0)$ may differ from $\bP_1$, 
and $\bP(0,-1)$ may differ from $\bP_2$ since $\bP(s_1,s_2)$ retains only 
the coefficients that are both non zero in $\bP_1$ and $\bP_2$. For example, 
$\bP(-1,0)$ and $\bP(0,-1)$ may be conjugate while $\bP_1$ and $\bP_2$ are not. 
Indeed we can have $\bP(-1,0)=\bP(0,-1)$ even when $\bP_1 \neq\bP_2$. 

In this section we first prove Theorem~\ref{theo11} which consider the case when
$\bP(-1,0)$ and $\bP(0,-1)$ are conjugate. Then we present
a detailed proof of Theorem~\ref{th3}, and finally we
briefly discussed proofs of Theorems~\ref{theogenmark}-- \ref{theogege}.

\begin{proof}[Proof of Theorem~\ref{theo11}]
When $\bP(-1,0)$ and $\bP(0,-1)$ are conjugate, the situation is more
similar to the case when $\bP-1=\bP_2$ discussed in Section~\ref{sec-preliminary}.
In particular there is no saddle point, and the analysis reduces to 
computing some residues of poles.

To start, we notice that in this case  there exists a vector of 
real numbers $(x_a)_{a\in\CA}$ such that
$$
P_1(a|b)P_2(a|b)>0\Rightarrow P_2(a|b)\frac{x_a}{x_b}P_1(a|b).
$$
As a consequence we have the same spectrum of $\bP(c-s,s)$ for all $c$ and $s$, 
and thus we have the identity $\lambda(c-s,s)=\lambda(c,0)$. We will prove the 
result when $\forall a\in\CA: x_a=1$. This does not 
implies that $\bP_1=\bP_2$ since the above identity only applies to nonzero 
coefficients in both matrices. In fact we only have $\bP(-1,0)=\bP(0,-1)$. 
We also notice that $\bP(s,1-s)$ is identical for all $s$. We leave as an exercise 
the case where the $x_a$  are not identical. We notice that $\CK$ consists 
of the tuple $(s,1-s)$ where $s$ is real.
We know that 
\be
C(z_1,z_2)=(1-e^{-z_1})(1-e^{-z_2})+\sum_{a\in\CA}C_a(\pi_1(a)z_1,\pi_2(a)z_2).
\label{CCz}\ee
The issue here is that the $\pi_1(a)$ are not necessarily equal to 
the $\pi_2(a)$ because they also depend on the other coefficients 
of matrices $\bP_1$ and $\bP_2$ which are not tied up by the conjugation 
property (because their {\it alter ego} coefficients in the other matrix are null). 
This implies that we have to consider a new matrix $\bG(z)$ whose coefficients 
$g_{a,b}(z)$ for $(a,b)\in\CA^2$  are
$$g_{a,b}(z)=C_a(\pi_1(b)z,\pi_2(b)z).$$
Let $\bP=\bP(-1,0)=\bP(0,-1)$, {\it i.e.} the matrix whose coefficients 
are those $\bP_1(a|b)$ when $P_1(a|b)P_2(a|b)>0$, and zero otherwise. 

 From the functional equation
$$
C_a(z_1,z_2)=(1-e^{-z_1})(1-e^{-z_2})+\sum_{c\in\CA}C_c(P(a|c)z_1,P(a|c)z_2).
$$
We have the identity
\be
g_{a,b}(z)=(1-e^{-\pi_1(b)z})(1-e^{-\pi_2(b)z})+\sum_{c\in\CA}g_{c,b}(P(a|c)z).
\label{gab}\ee
We have
$$
\trace(\bG(z))=\sum_{a\in\CA}C_a(\pi_1(a)z,\pi_2(a)z),
$$
and then we can rewrite equation~(\ref{CCz}) for $z_1=z_2=z$ as
$$
C(z,z)=(1-e^{-z})^2+\trace(\bG(z)).
$$
We then compute the Mellin transform $\bG^*(s)$ of $\bG(z)$, which is
a matrix of elements $g_{a,b}^*(s)$ (i.e., Mellin transforms of $g_{a,b}(z)$)
that are equal to
$$
g_{a,b}^*(s)=r_a(s)\Gamma(s)+\sum_{c\in\CA}g_{c,b}^*(s)P(a|c)^{-s}.
$$
Here 
$$r_a(s)=-\pi_1(a)^{-s}-\pi_2(a)^{-s}+(\pi_1(a)+\pi_2(a))^{-s}$$ 
is the Mellin transform of the term $(1-e^{-\pi_1(b)z})(1-e^{-\pi_2(b)z})$ 
in~(\ref{gab}). Equivalently
\be
\bG^*(s)=\Gamma(s)\br(s)\otimes\bone+\bP(s,0)\bG^*(s)
\ee
or $\bG(s)=(\bI-P(s,0))^{-1}\Gamma(s) \br(s)\otimes\bone$ where $\br(s)$ is the 
vector made of the $r_a(s)$'s. 
The Mellin transform of $C(z,z)$ which we denote as $c^*(s)$ satisfies
$$
c^*(s)=(2^{-s}-2)\Gamma(s)+\trace(\bG^*(s)).
$$

The first singularity of $\bG(s)$ is the pole of $(\bI-\bP(s,0))^{-1}$ 
which is at $s=-\kappa$ such that $\lb(s,0)=1$.  The only singular term of 
$\bG^*(s)$  at $s=-\kappa$ is on its main eigenvectors: 
$$\frac{\Gamma(s)}{1-\lb(s,0)}\left(\bzeta(s,0)\otimes \bu(s,0)\right)
\left(\br(s)\otimes \bone\right)$$
whose trace is
$$
\frac{\Gamma(s)}{1-\lb(s,0)}  \langle\br(s)|\bu(s,0)\rangle\langle\bzeta(s,0)|\bone\rangle.
$$

Thus the residue of $\trace(\bG^*(s))$ is
$$
\gamma_0(-\kappa)=\frac{1}{\lambda'(-\kappa,0)}\langle\br(-\kappa)|\bu(-\kappa,0)
\rangle\langle\bzeta(-\kappa,0)|\bone\rangle.
$$
The inverse Mellin gives
\be
C(z,z)=(1-e^{-z})^2+z^\kappa\gamma_0(-\kappa)\left(1+o(1)\right).
\ee
When $\bP(-1,0)$ is logarithmically rationally related, there are 
several poles of $(\bI-\bP(s,0))$ regularly spaced on the vertical 
axis $\Re(s)=-\kappa$ giving a periodic contribution $z^\kappa Q_0(\log z)$.
\end{proof}

\begin{proof}[Proof of theorem~\ref{th3}]
Let $H_1(s_1,s_2)=\frac{\partial}{\partial s_1}\lb(s_1,s_2)$ 
and
$$
f(s_1,s_2)=\la\bpi(s_1,s_2)|\bun(s_1,s_2)\ra, \ \ \ \  
g(s_1,s_2)=\la\bzeta(s_1,s_2)|\bone\ra. 
$$

We first notice that Lemma~\ref{lemLs} about the convexity of 
$L_1(s)$ is still valid since it depends only on general properties of the set $\CK$.
Define now $L_{j,k}(s)$ implicitly as
$$
\lambda_j(-L_{j,k}(s),s)=1.
$$ 
where $\lb_j(s_1,s_2)$ is the $j$th eigenvalues of matrix $\bP(s_1,s_2)$, listed
in decreasing modulus. The index $k$ indicates that these functions 
can be polymorphic since the root of the equation 
$\lb_j(s_1,s_2)=1$ for $s_2$ fixed can be multiple as we have seen 
in the case when one source is uniform memoryless. 
For $j$ fixed, each of the functions $L_{j,k}(s)$ are homeomorphic 
as long as $\lb_j(L_{j,k}(s),s)$ is non ambiguous {\it i.e.} 
the $j$th eigenvalue has not the same modulus as the previous or next eigenvalues. 
This would happen only on a discrete set of values $s$. 

Let now for $1\leq j\leq |\CA|$ and $k\in\mathbb{Z}$ 
\begin{eqnarray}
f_{j,k}(s)&=&\langle\bpi(-L_{j,k}(s),s)|\bu_j(-L_{j,k}(s),s)\rangle ,\\
g_{j,k}(s)&=&\langle\bzeta_j(-L_{j,k}(s),s)|\bone \rangle .
\end{eqnarray}
Then 
\begin{eqnarray}
&&I(z,\rho)=\\
&&\frac{1}{2i\pi}\int_{\Re(s)=\rho}\sum_{k\in\mathbb{Z}}
\sum_{j=1}^{|\CA|}\frac{f_{j,k}(s)g_{j,k}(s)\Gamma(-L_{j,k}(s))\Gamma(s)}
{-\frac{\partial}{\partial s_1}\lb_j(-L_{j,k}(s),s)}z^{L_{j,k}(s)-s}ds\nonumber
\end{eqnarray}
and 
\be
J(z,\rho)=\frac{1}{2i\pi}\int_{\Re(s)=c_2}\hC(0,s)\Gamma(s)z^{-s}ds.
\ee
With these new definitions the expression~(\ref{eq-rho}) in Lemma~\ref{lem-rho} 
is still valid, that is for all $M>0$: 
$$
C(z,z)=1+I(z,\rho)+J(z,\rho)+O(z^{-M}).
$$
The proof is indeed the same, the pole cancelations occur the same way and 
the identity between residues is formally the same. 

We know that for all $k$ we have $\lambda_1(-L_{1,k}(s),s)=1$. We assume that $k=0$ 
defines the branch where $L_{1,0}(s)$ is real when $s$ is real. To simplify
we denote $L(s)=L_{1,0}(s)$.  
We then move the integration path to the value $s=c_2$ which attains the 
minimum of $L(s)-s$ at $s=\kappa$. We know that the 
$$J(z,c_2)=O(z^{\kappa-\eps})$$
for $\eps>0$ such that $\kappa-\eps>L(0)$. Therefore 
$$C(z,z)=1+I(z,c_2)+O(z^{\kappa-\eps})$$
as in the case with uniform one source. Since $z^{L(s)-s}$ is at a minimum 
for real values of $s$, we have again a saddle point for 
$\int_{\Re(s)=c_2}\mu(s)z^{L(s)-s}ds$. 

Thus we arrive to two cases:
(i) either $c_1>0$ or $c_2>0$ or (ii)  $c_1$ and $c_2$ are both negative.
In the first case the condition of Theorem~\ref{th3} applies. In the second case the condition of Theorem~\ref{theogenmark} applies.
We consider the case $c_2>0$ (case $c_1>0$ is symmetric). 
Moving the integration toward $\Re(s)=c_2$ one meets the pole of function 
$\Gamma(s)$ at zero. 

When we meet the pole at $s=0$ by moving $\rho$ toward the positive value 
we obtain a residue from $I(z,\rho)$ equal to $\sum_j H_j(0,z)$ and a residue 
from $J(z,\rho)$ equal to $\hC(0,0)$. Notice that $H_1(0,z)=\Omega(z^{c_0})$ 
since $c_0=L(0)>0$ while the residue from $J(z,\rho)$ is 
negligible. The function $H_1(0,z)$ turns out to be the leading term since 
the other terms are of order $z^{L_j(0)}$ for $j>1$ and 
$\Re(L_j(0))<L_1(0)$. By moving again the integration path with $\Re(s)>0$ 
we arrive at $\Re(s)=c_2$ thus 
$$C(z,z)=\sum_j H_j(0,z)+I(z,c_2)+O(z^{-M}).$$ 

We know from the previous discussion that 
$I(z,c_2)=O(z^\kappa)$ which is of order smaller than $z^{L(0)}$ 
per definition of $\kappa$. Therefore we have
$C(z,z)=1+H_1(0,z)+O(z^{c_0-\eps})$ for some $\eps>0$. 
Notice that 
$$H_1(0,z)=\frac{f_1(0)g_1(0)}{\lambda_1(0)}z^{c_0}+Q_1(\log z)z^{c_0}
$$
where $Q_1(.)$ is a periodic function of periodic  $\log|\CA|$ and 
of mean 0 with small amplitude. 

Recapitulating all cases of Theorem~\ref{th3}:

(i) when $\bP(-1,0)$ is not logarithmically related, then
$$
C_{n,n}=\gamma_1(c_0,0)n^{-c_0}(1+o(1))
$$
with 
$$
\gamma_1(s_1,s_2)=\frac{f(s_1,s_2)g(s_1,s_2)(1+s_1)\Gamma(s_1)}{H_1(s_1,s_2)}.
$$
\parsec
(ii) When $\bP(-1,0)$ is logarithmically commensurable
$$
C_{n,n}=
n^{-c_0}\sum_{k\in\mathbb{Z}}n^{2ik\pi\nu}\gamma(c_0+2ik\pi\nu,0)
+O(n^{c_0-\eps})
$$
where $\nu$ is the root of $\bP_1$.
\end{proof}

\paragraph*{Remark} Remember that when all coefficients of $\bP_2$ 
are non-negative $\bP(-1,0)=\bP_1$ which is not necessarily the case when 
$\bP_2$ has some null coefficients. 

\begin{proof}[Proof of Theorems~\ref{theogenmark} and~\ref{theogege}]
We need the following lemma which is basically equivalent 
of Lemmas~\ref{lem9} and~\ref{lem18} developed in the special case. 

\begin{lemma}
There exists $\eps>0$ such that for all integers $j>1$ and for all 
integers $k$ the quantities
\be
\int_{\Re(s)=c_2}\frac{f_{j,k}(s)g_{j,k}(s)\Gamma(-L_{j,k}(s))\Gamma(s)}{-\frac{\partial}{\partial s_1}\lambda_j(-L_{j,k}(s),s)}z^{L_{j,k}(s)-s}ds
\ee
are uniformly $O(z^{\kappa-\eps})$. 
\end{lemma}
\begin{proof}
The proof consists of showing that there exist $\eps>0$ such that if 
$\lambda_j(s_1,s_2)=1$ with $\Re(s_2)=c_2$ then $\Re(s_1)>c_1+\eps$. 
First we prove that $\Re(s_1)\ge c_1$. We know that 
$|\lb_1(s_1,s_2)|\ge|\lb_j(s_1,s_2)|=1$. Since 
$|\lb_1(s_1,s_2)|\le\lb_1(\Re(s_1),\Re(s_2))$. We have $\Re(s_2)=c_2$, 
thus the inequality $\lb_1(\Re(s_1),c_2)\ge 1$ implies that 
$\Re(s_1)\ge c_1$ since $\lb_1(\Re(s_1),\Re(s_2))$ is strictly 
increasing in $\Re(s_1)$ and $\Re(s_2)$. 

Second we prove the existence of $\eps$. By absurdum we assume that there 
is a sequence of complex numbers $(x_k,y_k)$ such that 
$\lb_j(x_k,y_k)=1$ and $\Re(y_k)=c_2$ and $\Re(x_k)\to c_1$ with $\Re(x_k)\ge c_1$. 
We know that $|\lb_1(x_k,y_k)|\ge|\lb_j(x_k,y_k)|=1$. From the inequality 
$$\lb_1(\Re(x_k),\Re(y_k))\ge|\lb_1(x_k,y_k)|\ge 1$$ we get that 
$|\lb(x_k,y_k)|\to 1$, since $(\Re(x_k),\Re(y_k))\to(c_1,c_2)$. This would 
imply that $$\left|\frac{\lb_j(x_k,y_k)}{\lb_1(x_k,y_k)}\right|\to 1,$$ 
which contradicts Lemma~\ref{lem8}.
\end{proof}

The above lemma fills the gap necessary to establish 
of Theorems~\ref{theogenmark} and~\ref{theogege}
by following the footsteps of the proofs of 
Theorems~\ref{theogenmark0} and~\ref{theo01}.
\end{proof}

\section*{Appendix}


\begin{proof}[Proof of Lemma~\ref{lemc2}]
Let $\forall(a,b)\in\CA^2$ $P_1(a|b)>0$. Thus $(0,-1)\in\CK$, 
but since some $P_2(a|b)$ may be zero, the point $(-1,0)$ may not be $\CK$. 
But if $(-1,s_2)\in\CK$ with $s_2\in\mathbb{R}$  then $s_2\ge 0$, 
otherwise since $\bP(-1,s_2)\le\bP_1$ coefficientwise,  
then $\lb(-1,s_2)<1$. 
Similarly if $(s_1,0)\in\CK$ then $s_1\ge -1$.

We know that the curve $(-s_1,-s_2)$ is convex for $(s_1,s_2)\in\CK$, 
so is the curve $(s_1,-s_1-s_2)$, then  $-s_1-s_2$ is a function of 
$s_1$, say $a(s_1)$. We have $a(0)=1$ and $a(-1)\le 1$. Thus the 
minimum value of $a(s_1)$ which is $\kappa$ is attained on $c_1$ 
which must satisfies $c_1\le 0$ We then have $c_1<0$ 
when the curve is strictly convex. 

Similarly, the curve $(s_2,-s_1-s_2)$ is convex, so is the function 
$b(s_2)=-s_1-s_2$. Since $b(-1)=1$ and $b(0)\le 1$ the minimum is 
necessarily attained on $c_2\ge -1$, and $c_2>-1$ when it is strictly convex.
\end{proof}

\med
\begin{proof}[ Proof of Lemma~\ref{lemDePo}]
In order to prove Lemma~\ref{lemDePo} we adopt here the following general 
double depoissonization lemma that is proved in \cite{js-book} 
(see Lemmma 10.3.4 in Chapter  10).

\begin{lemma}
Let $a_{n,m}$ be a two-dimensional (double) sequence of complex numbers.
We define the double Poisson transform $f(z_1,z_2)$ of $a_{n,m}$ as
$$
f(z_1,z_2)=\sum_{n,m\ge 0} a_{n,m}
\frac{z_1^n}{n!}\frac{z_1^m}{m!}e^{-z_1-z_2}.
$$
Let now $\S_\th$ be a cone of angle $\theta$ around the real axis.
Assume that there exist $B>0$, $D>0$, $\alpha<1$ and $\beta$ such that
for $|z_1|$, $|z_2|\to\infty$:
\begin{itemize}
\item[{\rm (i)}] if $z_1,z_2\in\S_\th$ then $|f(z_1,z_2)|=
B (|z_1|^\beta+|z_2|^\beta)$;
\item[{\rm (ii)}] if $z_1,z_2\notin\S_\th$ then $|f(z_1,z_2)e^{z_1+z_2}|=
De^{\alpha|z_1|+\alpha|z_2|}$;
\item[{\rm (iii)}] if $z_i\in\S_\th$ and $z_j\notin\S_\th$ for
$\{i,j\}=\{1,2\}$ and $~|f(z_1,z_2)e^{z_j}|<D|z_i|^\beta e^{\alpha|z_j|}.$
\end{itemize}
Then 
$$
a_{n,m}=f(n,m)+O\left(\frac{n^\beta}{m}+\frac{m^\beta}{n}\right)
$$
for large $m$ and $n$.
\label{coro1}
\end{lemma}

Just to prove the Lemma~\ref{lemDePo}, we need to establish three 
conditions (i)-(iii) of Lemma~\ref{coro1}.
We accomplish it through a generalization of the  so called
{\it increasing domain} approach discussed in \cite{js98,spa-book}.

We first prove the lemma for the generating functions
$C_a(z_1,z_2)$ for every $a\in\CA$.
Assume now that
$\rho=\max_{(a,b)\in\CA^2,i\in\{1,2\}}\{P_i(a|b)\}$. We denote by $\S_k$ part of the
cone $\S_\th$ that contains points such that $|z|<\rho^{-k}$.
Notice that $\S_k\subset\S_{k+1}$ for all integer $k$.
We also notice $C(z_1,z_2)=O((|z_1|+|z_2|)^2)$ when $z_1,z_2\to 0$,
therefore we can define
$$
B_k=\max_{a\in\CA,(z_1,z_2)\in\S_k\times\S_k}\frac{|C_a(z_1,z_2)|}{|z_1|+|z_2|}
<\infty~.
$$
We use the functional equation 
\be
C_b(z_1,z_2)=(1-(1+z_1)e^{-z_1})(1-(1+z_2)e^{-z_2})+\sum_{a\in\CA}C_a
\left(P_1(a|b) z_1,P_2(a|b)z_2\right).
\ee
In the above equation, we notice that if $(z_1,z_2)\in\S_{k+1}\times\S_{k+1}-
\S_k\times\S_k$, then
for all $(a,b)\in\CA^2$  $(P_1(a|b) z_1,P_2(a|b) z_2)$ are in
$\S_k\times\S_k$ and therefore we have for some fixed $\beta>0$ and for all $b\in\CA$:
\be
|C_b(z_1,z_2)|\le B_k(\sum_{a\in\CA}P_1(a|b)|z_1|+P_2(a|b)|z_2|)+\beta=
B_k(|z_1|+|z_2|)+\beta
\ee
since $|1-(1+z_i)e^{-z_i}|$ is uniformly bounded for all integers $k$ by
some $\sqrt{\beta}$ for both $i\in\{1,2\}$ when $(z_1,z_2)\in\S_k$.
Thus, we can derive the following recurrent inequality:
\be
B_{k+1}\le B_k+\beta\max_{(z_1,z_2)\in\S_{k+1}\times\S_{k+1}-\S_k\times\S_k}
\{\frac{1}{|z_1|+|z_2|}\}=B_k+\beta\rho^k~.
\ee
We should notice that
\be
\min_{(z_1,z_2)\in\S_{k+1}\times\S_{k+1}-\S_k\times\S_k}\{|z_1|+|z_2|\}=\rho^{-k}
\ee
because one of the number $z_i$ has modulus greater than $\rho^{-k}$.
It turns out that $\lim_{k\to\infty} B_k<\infty$, establishing
condition (i) of the double depoissonization Lemma~\ref{coro1}.

Now we are going to establish condition (iii). To this end we define $\CG$ as the
complementary cone of $\S_\th$ and $\CG_k$ as the portion made of the
point of modulus smaller than $\rho^{-k}$. We will use $\cos\theta<\alpha<1$,
therefore $\forall z\in\CG$: $|e^z|<e^{\alpha|z|}$.
We define $D_k$ as
\be
D_k=\max_{a\in\CA,(z_1,z_2)\in\CG_k\times\CG_k}\frac{|C_a(z_1,z_2)e^{z_1+z_2}|}
{\exp(\alpha|z_1|+\alpha |z_2|)}~.
\ee
We define $G_a(z_1,z_2)=C_a(z_1,z_2)e^{z_1+z_2}$, we have the following equation
\be
G_b(z_1,z_2)=(e^{z_1}-1-z_1)(e^{z_2}-1-z_2)+\sum_{a\in\CA}C_a
\left(P_1(a|b) z_1,P_2(a|b)z_2\right)e^{1-P_1(a|b)z_1+(1-P_2(a|b))z_2}~.
\ee
We notice that if $(z_1,z_2)\in\CG_{k+1}\times\CG_{k+1}-\CG_k\times\CG_k$, then
all  $(P_1(a|b) z_1,P_2(a|b) z_2)$ are in $\CG_k\times\CG_k$
and therefore we have for all $b\in\CA$:
\begin{eqnarray*}
|G_b(z_1,z_2)|&\le& D_k\left(\sum_{a\in\CA}\exp\left( (P_1(a|b)\alpha+(1-P_1(a|b))
\cos\theta )|z_1|
+(P_2(a|b)\alpha+(1-P_2(a|b))\cos\theta)|z_2|\right)\right)\\
&&+(e^{\cos\theta|z_1|}+1+|z_1|)(e^{\cos\theta|z_2|}+1+|z_2|).
\end{eqnarray*}
We notice that $\forall (a,b)\in\CA^2$ and $\forall i\in\{1,2\}$:
\be
P_i(a|b)\alpha+(1-P_i(a|b))\cos\theta-\alpha\le-(1-\rho)(\alpha-\cos\theta)~,
\ee
We also have $e^{\cos\theta|z_i|}+1+|z_i|\le e^{\cos\theta|z_i|}(2+
\frac{1}{e\cos\theta})$, therefore
\be
\frac{|G_b(z_1,z_2)|}{\exp(\alpha(|z_1|+|z_2|))}\le D_k|\CA|e^{-(1-\rho)
(\alpha-\cos\theta)(|z_1|+|z_2|)}+(2+\frac{1}{e\cos\theta})^2
e^{-(\alpha-\cos\theta)(|z_1|+|z_2|)}~.
\ee
Since $(z_1,z_2)\in\CG_{k+1}\times\CG_{k+1}-\CG_k\times\CG_k$ implies
$|z_1|+|z_2|\ge\rho^{-k}$ it  follows
\be
D_{k+1}\le \max\left\{D_k,|\CA|D_k e^{-(1-\rho)(\alpha-\cos\theta)\rho^{-k}}+
(2+\frac{1}{e\cos\theta})^2 e^{-(\alpha-\cos\theta)\rho^{-k}}\right\}~.
\ee
We clearly have $\lim_{k\to\infty} D_k<\infty$ and condition (iii) is established.

The proof of condition (ii) for $z_1$ and $z_2$  being
in $\S_\th$ and $\CG$ is a mixture of
the above proofs.
Furthermore, the proof about the unconditional generating function $C(z_1,z_2)$ is a
trivial extension. 
\end{proof}

\med
\begin{proof}[ Proof of Lemma~\ref{lem-conj}]
Let $\bun=(u_a)_{a\in\CA}$ be the right eigenvector of $\bM$ and
$(v_a)_{a\in\CA}$ be the right eigenvector of $\bQ$. Let also $v_a=x_a u_a$. 
If 1 is the eigenvalue, we have for all $c\in\CA$:
\be
(1-e^{i\theta_{cc}}m_{cc})u_c=\sum_{b\neq c}m_{cb}u_be^{i\theta_{cb}}
\frac{x_b}{x_c}~.
\ee
If $e^{i\theta_{cc}}\neq 1$, then
\be
|(1-e^{i\theta_{cc}}m_{cc})u_c|>(1-m_{cc})u_c.
\ee
By the Perron-Frobenius theorem all $u_a$ are real non negative.
Suppose that $|x_c|=\max_{a\in\CA}\{|x_a|\}$.
If $\exists d\in\CA$: $\frac{|x_d|}{|x_c|}<1$ or if $(b,b')\in(\CA-\{c\})^2$: 
$e^{i\theta_{cb}}\frac{x_b}{x_c}\neq e^{i\theta_{cb'}}\frac{x_b'}{x_c}$. Then
\be
\left|\sum_{b\neq c}m_{cb}u_be^{i\theta_{cb}}\frac{x_b}{x_c} 
\right|<\sum_{b\neq c}m_{cb}u_b~.
\ee
But we also know that 
\be
(1-m_{cc})u_c=\sum_{b\neq c}m_{cb}u_b~.
\ee
Therefore, we have 
$e^{i\theta_{cc}}=1$ and for all $b\in\CA$: $|x_b|=|x_c|$, and for all
$(b,b')\in(\CA-\{c\})^2$: $e^{i\theta_{cb}}\frac{x_b}{x_c}= 
e^{i\theta_{cb'}}\frac{x_b'}{x_c}$. But since for all $b\in\CA$ $|x_b|=|x_c|$ 
every symbol in $\CA$ can play the role of $c$. Since for all $c\in\CA$
\be
(1-m_{cc})=\sum_{b\neq c}m_{cb}u_be^{i\theta_{cb}}\frac{x_b}{x_c}=
\sum_{b\neq c}m_{cb}u_b~,
\ee
we simply have $\forall (a,b)\in\CA$: $e^{i\theta_{ab}}\frac{x_b}{x_a}=1$.
Denoting $x_a=e^{i\theta_a}$ we prove the expected result. 
The converse proposition is immediate.
\end{proof}
\med
\begin{proof}[Proof of Lemma~\ref{lemconc} and~\ref{lemconv}]
We call $\tilde{\CK}$ the set of real tuples such that $\lb(s_1,s_2)\le 1$. The set
$\CK$ is the topological border of $\tilde{\CK}$ and since $\lb(s_1,s_2)$ decreases when
$s_1$ or $s_2$ decrease, it is the upper border. We will show that
$\tilde{\CK}$ is a convex set and thus its upper border is concave.
Let $(x_1,x_2)$ and $(y_1,y_2)$ be two elements of $\tilde{\CK}$ and $\alpha$
and $\beta$ two non negative real numbers such that $\alpha+\beta=1$.
We want to prove that $(\alpha x_1+\beta y_1,\alpha x_2+\beta y_2)\in\tilde{\CK}$.

By construction
$$
\bP(\alpha x_1+\beta y_1,\alpha x_2+\beta y_2)=\bP(\alpha x_1,\alpha x_2)
\star\bP(\beta y_1,\beta y_2)
$$
where $\star$ denotes the Schur product.
For $(s_1,s_2)\in\tilde{\CK}$ let $\bun(s_1,s_2)$ the right main
eigenvector of $\bP(s_1,s_2)$, {\it i.e.}
$$\bP(s_1,s_2)\bun(s_1,s_2)=\lb(s_1,s_2)\bun(s_1,s_2).$$
We know that $\lb(s_1,s_2)\le 1$ therefore
$$
\bP(s_1,s_2)\bun(s_1,s_2)\le\bun(s_1,s_2)
$$
coefficientwise. Let $\bun(s_1,s_2)^{\star\alpha}$ denotes
the vector $\bun(s_1,s_2)$ with all its coefficients raised to
power $\alpha$. We want to give an estimate of
$$\bP(\alpha x_1,\alpha x_2)\star\bP(\beta y_1,\beta y_2)
$$
applied to the vector
$$
\bun(x_1,x_2)^{\star\alpha}\star\bun(y_1,y_2)^{\star\beta}.
$$
 Let $a\in\CA$ the coefficient of the vector
$$\bP(\alpha x_1,\alpha x_2)\star\bP(\beta y_1,\beta y_2)
\bun(x_1,x_2)^{\star\alpha}\star\bun(y_1,y_2)^{\star\beta}
$$
corresponding to symbol $a$ is equal to
$$
\sum_{b\in\CA}u_b(x_1,x_2)^{\alpha}P_1(a|b)^{-\alpha x_1}
P_2(a|b)^{-\alpha x_2}u_b(y_1,y_2)^{\beta}P_1(a|b)^{-\beta y_1}
$$
$$P_2(a|b)^{-\beta y_2}
\sum_{b\in\CA}u_b(x_1,x_2)^{\alpha}P_1(a|b)^{-\alpha x_1}
P_2(a|b)^{-\alpha x_2}u_b(y_1,y_2)^{\beta}P_1(a|b)^{-\beta y_1}P_2(a|b)^{-\beta y_2}.
$$
Using H{\"o}lder inequality, the above quantity is smaller than
\be
\left(\sum_{b\in\CA}u_b(x_1,x_2)P_1(a|b)^{-x_1}P_2(a|b)^{-x_2}\right)^\alpha
\left(\sum_{b'\in\CA}u_b(y_1,y_2)P_1(a|b')^{-y_1}P_2(a|b')^{-y_2}\right)^\beta.
\ee
The above terms are 
respectively 
$\lb(x_1,x_2) u_a(x_1,x_2)$ and $\lb(y_1,y_2) u_a(y_1,y_2)$.
Therefore the vector
$$\bP(\alpha x_1,\alpha x_2)\star\bP(\beta y_1,\beta y_2)
\bun(x_1,x_2)^{\star\alpha}\star\bun(y_1,y_2)^{\star\beta}
$$
is coefficientwise smaller than
$$\lb^{\alpha}(x_1,x_2)\lb^\beta(y_1,y_2)\bun(x_1,x_2)^{\star\alpha}
\star\bun(y_1,y_2)^{\star\beta}.
$$
Since $\lb^{\alpha}(x_1,x_2)\lb(^\beta(y_1,y_2)\le 1$ by
Perron-Frobenius the main eigenvalue of
$\bP(\alpha x_1+\beta y_1,\alpha x_2+\beta y_2)$ is smaller than or equal to 1,
consequently $(\alpha x_1+\beta y_1,\alpha x_2+\beta y_2)\in\bar{\CK}$.

The H{\"o}lder inequality is an equality if and only if the
vectors $(u_a(x_1,x_2)\bP(x_1,x_2))_{a\in\CA}$ and
$(u_a(y_1,y_2)\bP(y_1,y_2))_{a\in\CA}$ are colinear, which happens
when $\bP(x_1,x_2)$ and $\bP(y_1,y_2)$ are conjugate, which is
equivalent to the fact that $\bP_1$ and $\bP_2$ are conjugate
(on the coefficients which are non zero).
\end{proof}
\med
\begin{proof}[ Proof of Lemma~\ref{lem8}]
Consider the matrix $\frac{1}{\lb(x_k,y_k)}\bP(x_k,y_k)$. Since the 
coefficients of this matrix are bounded, there is no loss in generality 
to consider the sequence of matrices converging to a matrix $\bM$. 
The matrix $\bM$ and matrix $\bQ=\bP(c_1,c_2)$, as defined 
in Lemma~\ref{lem-conj}, are imaginary conjugate {\it i.e.} 
the coefficients of $\bM$ are of the form 
\be
e^{i(\theta_a-\theta_b)}P_1(a|b)^{-c_1}P_2(a|b)^{-c_2}
\ee
for some vector of real numbers $\theta_a$. Therefore, $\bM$ and 
$\bP(c_1,c_2)$ have the same spectrum. The 
spectrum of $\frac{1}{\lb(x_k,y_k)}\bP(x_k,y_k)$ converges to the spectrum 
of $\bM$. 
Furthermore, the right eigenvector $\bun(x_k,y_k)$ converges 
to the vector $e^{i\theta_a}u_a(c_1,c_2)$ and the left eigenvector 
$\bzeta(x_k,y_k)$ converges to $e^{-i\theta_a}\zeta_a(c_1,c_2)$. 

For any pair of complex numbers $(s_1,s_2)$ we have the identity
\be
\frac{1}{\lb(x_k,y_k)}\bP(x_k+s_1,y_k+s_2)=\frac{1}{\lb(x_k,y_k)}
\bP(x_k,y_k)*\bP(s_1,s_2)~.
\ee
Thus $\frac{1}{\lb(x_k,y_k)}\bP(x_k+s_1,y_k+s_2)$ converges to 
$\bM * \bP(s_1,s_2)$ and is 
conjugate to $\bP(c_1+s_1,c_2+s_2)$. 
Since the eigen spectrum of $\frac{1}{\lb(x_k,y_k)}\bP(x_k+s_1,y_k+s_2)$ 
converges to the eigen spectrum of  $\bP(c_1+s_1,c_2+s_2)$,  
thus we have $\lb(x_k+s_1,y_k+s_2)\to\lb(c_1+s_1,c_2+s_2)$. We also have 
$|\lb(x_k+s_1,y_k+s_2)|>|\lb_2(x_k+s_1,y_k+s_2)|$ when $k$ is large 
enough with $(s_1,s_2)$ in the complex neighborhood $\CU^2$ which implies 
the analyticity of $\lb(x_k+s_1,y_k+s_2)$.
Thus by Ascoli theorem the derivatives converge, too.  
\end{proof}


\begin{thebibliography}{99}

\bibitem{fgd}
P. Flajolet, X. Gourdon, and P. Dumas,
Mellin Transforms and Asymptotics: Harmonic sums,
{\it Theoretical Computer Science}, 144, 3--58, 1995.

\bibitem{bh12} V. Becher and P. A. Heiber, 
A better complexity of finite sequences, 
Abstracts of the 8th {\it Int. Conf. on Computability and Complexity in 
Analysis} and 6th {\it Int. Conf. on Computability, Complexity, 
and Randomness} , 
Cape Town, South Africa, January 31,  February 4, 2011, p. 7.

\bibitem{fs-book}
P. Flajolet and R. Sedgewick,
{\it Analytic Combinatorics}, Cambridge University Press, Cambridge, 2008.

\bibitem{IYZ02}
Ilie, L., Yu, S., and Zhang, K.
\newblock Repetition Complexity of Words
\newblock In {\em Proc. COCOON\/} 320--329, 2002.

\bibitem{jacquet}
P. Jacquet,
Common words between two random strings,
\newblock {\it IEEE Intl. Symposium on Information Theory},
1495-1499, 2007.


\bibitem{js94}
P. Jacquet, and W. Szpankowski,
Autocorrelation on Words and Its Applications. Analysis of Suffix Trees by
String-Ruler Approach,
{\it J. Combinatorial Theory Ser. A}, 66, 237--269, 1994.


\bibitem{js98}
P. Jacquet, and W. Szpankowski,
Analytical DePoissonization and Its Applications, {\it Theoretical
Computer Science}, 201, 1--62, 1998.

\bibitem{js12}
P. Jacquet and W. Szpankowski,
Joint String Complexity for Markov Sources,
{\it 23rd International Meeting on Probabilistic, Combinatorial and
Asymptotic Methods for the Analysis of Algorithms}, AofA'12,
{\it DMTCS Proc.}, 303-322, Montreal, 2012.

\bibitem{js-book}
P. Jacquet, and W. Szpankowski,
{\it Analytic Pattern Matching: From DNA to Twitter},
Cambridge University Press, Cambridge, 2015.

\bibitem{jst01}
P. Jacquet, W. Szpankowski, and J. Tang,
Average Profile of the Lempel-Ziv Parsing Scheme for a Markovian Source,
{\it Algorithmica}, 31, 318-360, 2001.

\bibitem{js16}
P. Jacquet and W. Szpankowski,
Average Size of a Suffix Tree for Markov Sources,
{\it  27th International Meeting on Probabilistic, Combinatorial and
Asymptotic Methods for the Analysis of Algorithms}, AofA'16, Krakow, 2016.


\bibitem{jls04}
S. Janson, S. Lonardi and W. Szpankowski,
On Average Sequence Complexity,
{\it Theoretical Computer Science}, 326, 213-227, 2004.

\bibitem{hj}
R.~A.~Horn and C.~R.~Johnson,
{\it Matrix Analysis}, Cambridge University Press, Cambridge, 1985.

\bibitem{ralph13}
K. Leckey, R. Neininger, and W. Szpankowski,
Towards More Realistic Probabilistic Models for Data Structures:
The External Path Length in Tries under the Markov Model,
{\it  SIAM-ACM Symposium on Discrete Algorithms} (SODA 2013),  877-886,
New Orleans, 2013.


\bibitem{Li93}
 Li, M., and Vitanyi, P.
\newblock {\em Introduction to Kolmogorov Complexity and its Applications.}
\newblock Springer-Verlag, Berlin, Aug. 1993.

\bibitem{ms13}
N. Merhav and W. Szpankowski,
Average Redundancy of the Shannon Code for Markov Sources,
{\it IEEE Trans. Information Theory}, 59, 7186-7193, 2013.


\bibitem{niederreiter}
Niederreiter, H.,
\newblock Some computable complexity measures for binary sequences,
\newblock In {\em Sequences and Their Applications}, Eds. C. Ding,
T. Hellseth and H. Niederreiter
\newblock Springer Verlag,  67-78, 1999.

\bibitem{fw05}
J. Fayolle, M. Ward,
\newblock Analysis of the average depth in a suffix tree under a Markov model
\newblock DMTCS Proceedings of AofA 2005.

\bibitem{pittel85}
B. Pittel,
Asymptotic Growth of a Class of Random Trees,
{\it Annals of Probability}, 18, 414--427, 1985.


\bibitem{rs98}
M. R\'{e}gnier and W. Szpankowski,
On pattern frequency occurrences in a Markovian sequence,
{\it Algorithmica}, 22, 631-649, 1998.

\bibitem{spa-book}
W. Szpankowski,
{\it Analysis of Algorithms on Sequences},
John Wiley, New York, 2001.

\bibitem{ziv}
J. Ziv, 
On classification with empirically observed statistics and universal data
compression, 
{\it IEEE Trans. Information Theory}, 34, 278-286, 1988.


\end{thebibliography}
\end{document}